\documentclass[]{article}
\usepackage{mathtools}
\mathtoolsset{mathic} 
\usepackage{amsthm}
\usepackage{amssymb}
\usepackage{textcomp}
\usepackage{enumitem}
\usepackage{comment}
\usepackage{centernot}
\setlist[enumerate]{leftmargin=*}
\setlist[itemize]{leftmargin=*}
\setlist[description]{font=\mdseries\textsf, leftmargin=1.5em}
\usepackage[numbered]{bookmark} 

\newcommand{\customqed}[1]{{\renewcommand{\qedsymbol}{#1}\qed}}
\newcommand{\varqed}{\customqed{\hbox{$\lrcorner$}}}

\theoremstyle{plain}

\newtheorem{lemma}{Lemma}[section]

\newtheorem{theorem}{Theorem}
\newtheorem{proposition}[lemma]{Proposition}

\newtheorem{claim}[lemma]{Claim} 

\theoremstyle{definition}

\newtheorem{Definition}[lemma]{Definition}
\newenvironment{definition}{\begin{Definition}}{\varqed\end{Definition}}

\newtheorem{Notation}[lemma]{Notation}

\newtheorem{Condition}[lemma]{Condition}

\newtheorem{Remark}[lemma]{Remark}
\newtheorem*{Remarknonumber}{Remark} %
\newtheorem{Remarks}[lemma]{Remarks}

\newenvironment{remark}{%
	\begin{Remark}}{\varqed\end{Remark}}

\newtheorem{Example}[lemma]{Example}
\newtheorem{Examples}[lemma]{Examples}

\newcommand{\bbN}{\mathbb{N}}
\newcommand{\binv}{\{0,1\}}
\newcommand{\strings}{\binv^*}
\newcommand{\lstrings}{\binv^\ell}
\newcommand{\bseqs}{\binv^\infty}
\newcommand{\prefixof}{\prec}
\newcommand{\prefixeqof}{\preceq}
\DeclareMathOperator{\dom}{dom}
\title{Betting strategies with bounded splits}
\author{Tomislav Petrovi\'c}

\begin{document}

\maketitle

\begin{abstract}
We show that a pair of Kolmogorov-Loveland betting strategies
cannot win on every non-Martin-L\"of random sequence if
either of the two following conditions is true:
\begin{enumerate}[label=(\Roman*)]
	\item\label{upper_bounded_splits} 
	There is an unbounded computable function $g$ such that
	both betting strategies, when betting on an infinite binary sequence,
	almost surely, for almost all $\ell$,
	bet on at most $\ell-g(\ell)$ positions among the first $\ell$ positions
	of the sequence.
	\item\label{lower_bounded_splits} 
	There is a sublinear function $g$ such that
	both betting strategies, when betting on an infinite binary sequence,
	almost surely, for almost all $\ell$,
	bet on at least $\ell-g(\ell)$ positions among the first $\ell$ positions
	of the sequence.
\end{enumerate}

\end{abstract}
\section{Introduction}
Whether Kolmogorov-Loveland randomness is equal to Martin-L\"of randomness is
a well known open question in the field of algorithmic randomness
\cite{Ambos},\cite{Open}.

A Kolmogorov-Loveland betting strategy, starting with a finite amount of capital, makes bets on the values of bits at positions of an infinite binary sequence.
The positions can be chosen adaptively, based on the values of bits at positions that the betting strategy has previously bet on.

Once a new position to bet on is chosen,
the Kolmogorov-Loveland betting strategy then guesses the value
of the bit of the sequence
at the position 
and
uses some fraction of its capital as a wager.
If the guess was correct, the wager is doubled, and the capital increases
by the wagered amount, and if not, the wager is lost,
and the capital decreases
by the wagered amount.
The betting strategy wins on a sequence if, in the succession of bets, 
the supremum of capital is unbounded.

A sequence is Kolmogorov-Loveland random (KLR)
if no partialy computable Kolmogorov-Loveland betting strategy (KLBS) wins on the sequence. In fact, we can consider only
computable KLBS-es, since for every partialy computable KLBS 
there is a pair of computable KLBS-es that wins on the same sequences
\cite{Merkle}.

In this paper we'll look at Martin-L\"of randomness only with respect to the uniform Lebesgue measure.

A constructive null cover is a computable sequence $f_n$
of computable sequences of finite binary strings
such that for all $n$, $\sum_{s\in f_n}2^{-|s|}\leq 2^{-n}$,
and for every string $s$ in the sequence $f_{n+1}$ there is
a prefix of $s$ in $f_n$.

The set of infinite binary sequences that,
for all $n$, have a prefix in $f_n$ is called an effective nullset.
A sequence is Martin-L\"of random (MLR) if it is not contained in any effective nullset.

There is another characterization of MLR sequences, in terms of 
turing machines.
We'll use the monotone Turing machine and monotone complexity as defined in the textbook \cite{LiVitanyi}.
A monotone turing machine is a Turing machine with 
an infinite input tape, 
an infinite output tape and 
an infinite work tape. 
In each step, a monotone turing machine either reads the next bit from the input tape, or writes the next bit on the output tape, or does a step of computation on the work tape.
The montone complexity of a string $s$ with respect to the monotone Turing machine $T$ is the length of the shortest sequence of bits that
need to be read from the input tape, before the string $s$ is written on the output tape. We denote it with $\text{Km}_T(s)$.

It can be shown, \cite{Levin_Km},\cite{Schnorr_Km},
that there is a (universal) monotone Turing machine $U$ such that for every other monotone Turing machine $T$ and every string $s$,
$\text{Km}_U(s)$ lower bounds $\text{Km}_T(s)$ up to some additive constant that depends on $T$.
We call $\text{Km}_U(s)$ the monotone complexity of the string $s$.
A sequence is MLR if and only if there is an additive constant such that
the monotone complexity of any prefix of the sequence is within the constant of the length of the prefix.

It is easy to see that a \textit{single} partialy computable KLBS cannot win on all non-MLR sequences.
Suppose the KLBS is not total computable.
Then there is some finite sequence of bets, and after the last bit was reveiled no further bets are made.
The supremum of capital in this sequence of bets is bounded,
and the set of infinite binary sequences consistent with the outcomes of bets
contains non-MLR sequences, for instance, the ones that end in infinitely many zeros.

On the other hand, suppose the KLBS is total computable.  
We can find an infinite sequence of bets where the betting strategy never increases capital (to wit, the loosing streak).
The set of infinite binary sequences, consistent with this sequence of betting outcomes, is an effective null set.

However, it is still not known if there is maybe just a pair of partialy computable KLBS-es 
such that on every non-MLR sequence at least one of them wins
\cite{Open}.

\subsection{Related results}
In \cite{Muchnik},
Kolmogorov-Loveland randomness was defined, however, there it is called
unpredictability.
They show a pair of KLBS-es that, 
given an unbounded computable function $g$,
win on any sequence that, for every $\ell$, has a prefix of length $\ell$ whose monotone complexity is upper bounded by $\ell-g(\ell)$.
(Theorem 9.1 in \cite{Muchnik})

In \cite{Merkle} a pair of KL betting strategies is defined, that, given a computable partition of
positions into two infinite sets, wins on any sequence whose sub-sequence on
positions in the first set, and sub-sequence on positions in the second set,
 are
both non-MLR. (12 Theorem in \cite{Merkle})

In \cite{Open} a more restrictive kind of betting strategy than KLBS is proposed,
the betting strategy has to bet on positions non-adaptively, in some computable order.
The proofs in \cite{Laurent}, \cite{Lempp}, \cite{Merkle} on permutation randomness,
and methods used in \cite{BPP_EXP} are similar to the ones we'll use
to prove item \ref{lower_bounded_splits} from the abstract (theorem \ref{th_lower_bound}).

\subsection{Main result}
We'll prove item \ref{upper_bounded_splits} from the abstract,
for a more general kind of betting strategy than KLBS.

To make a bet, a KLBS chooses a position, and places a wager on the value of the bit at that position. The outcome of the bet reveals the bit value the sequence has at the chosen position, and the capital is updated.
Note that choosing a position splits a set of sequences into two clopen sets,
the ones that have $0$ at the chosen position, and the ones that have $1$.

A general betting strategy splits a set of sequences into any two clopen sets $v_0,v_1$, and places a wager on one of those sets.
The outcome of the bet reveals in which of the two sets the sequence is in,
and the capital is updated.
Unlike for KLBS, the sets $v_0,v_1$ might have unequal uniform Lebesgue measure (size), and in case the general betting strategy correctly guesses
in which set the sequence is in, the wagered amount might be more (or less)
than doubled depending on the size of the set on which the wager was placed.
Namely, let $v$ be the set of sequences consistent with outcomes of previous bets, and
let $c$ be the capital after the last one.
We say that $v$ is a part of the betting strategy, and $c$ is its capital.

The betting strategy makes a new bet by splitting $v$ into $v_0,v_1$ and
places a wager $w\leq c$ on one of them (say $v_0$).
If the sequence is in part $v_1$, the wager is lost, and its capital is $c-w$. 
If the sequence is in $v_0$, the wager is increased by dividing it with the size of $v_0$, conditional on $v$. The capital of part $v_0$ is then $c-w+\frac{1}{\lambda(v_0|v)}w$.

For a sequence of bets, and positions $[1,\ell]$, we will look at the number of times the bets have split the set of sequences in a way that depends on the first $\ell$ bits of the sequence.
More precisely, for a "tail" sequence $\rho$,
we will count the number of times a bet was made that splits
a part
into two parts so that both contain a sequence that, after first $\ell$ positions, ends with $\rho$.
Let this number be $N$,
and let $\sigma$ be an infinite binary sequence that
is consistent with the sequence of betting outcomes,
and, 
after first $\ell$ positions ends with the tail $\rho$.
We'll say that $\sigma$ was $N$ times $[1,\ell]$-split on by the betting strategy.
Clearly, $N\leq 2^\ell$.

We'll say that a betting strategy has splits upper bounded by function $f$
if, almost surely, a sequence was  $[1,\ell]$-split on at most $f(\ell)$ many times, for almost all $\ell$.

We show that if a pair of general betting strategies has splits upper bounded by $\ell-\log\ell-g(\ell)$, where $g$ is computable and unbounded,
then there is a non-MLR sequence on which neither betting strategy wins
(theorem \ref{th_upper_bound}).

Note that a KLBS has the following property:
for a finite set of positions $I$,
and a sequence of binary values $\rho$,
if a part $v$ is split into two parts $v_0,v_1$ so that
both contain sequences that on positions outside $I$ have values $\rho$,
then the number of such sequences in $v_0$ and in $v_1$ is the same.
We'll say that betting strategies with this property are half-splitting.
Clearly, a sequence can be $[1,\ell]$-split on at most $\ell$ many times
by a half-splitting betting strategy.

We show that if a pair of half-splitting betting strategies has splits upper bounded by $\ell-g(\ell)$, where $g$ is computable and unbounded,
then there is a non-MLR sequence on which neither betting strategy wins
(theorem \ref{th_upper_bound_half-splitting}).

\section{Notation}
We denote the set of (finite binary) strings with $\strings$,
the empty string with $\Lambda$,
the strings of length $\ell$ with $\lstrings$, 
the set of infinite binary sequences with $\bseqs$. When there is no confusion, we abbreviate infinite binary sequence to sequence.
We denote that $s$ is a prefix of a string (or a sequence) $s'$ with $s\prefixof s'$. The length of a string $s$ is denoted with $|s|$.
The set of all sequences prefixed by a string $s$ is denoted with $[s]$, such sets are called basic sets. 
A union of possibly infinitely many basic sets is called an open set.
The complement of an open set is called a closed set. 
A union of finitely many basic sets is called a clopen set.

We'll use \textit{size} as a shorthand for the uniform Lebesgue measure.
We'll denote with $\lambda$  the size of a set of sequences (e.g. $\lambda([s])=2^{-|s|}$).

The set of natural numbers is denoted with $\bbN$, $\bbN^k$ denotes the set of sequences of natural numbers of length $k$, and $\bbN^*$ the set of finite sequences of natural numbers. The empty sequence of natural numbers is denoted with $\_$.

To stress that we're taking a union of disjoint sets, we use $\sqcup$ instead of $\cup$.

\begin{definition}
	A map from a subset of $\bbN$ to binary values is called a
	\textit{restriction}.
	For a restriction with domain $I\subseteq\bbN$, we'll say 
	that the restriction \textit{restricts $I$}, and that \textit{$I$ are positions restricted by} the restriction. 
	A restriction with an empty domain is called the \textit{empty} restriction.
	A restriction with a finite domain is a \textit{finite} restriction.
	We denote the set of all restrictions that restrict
	$I$ with $\{0,1\}^I$.
	
	We'll say that a sequence \textit{is consistent with} the restriction $r\in\{0,1\}^I$ when
	the sequence has the same binary values at positions in $I$ as the restriction $r$. We denote with $[r]$ the set of sequences consistent with $r$.
	
	If two restrictions map the same position to different binary values,
	we'll say they are \textit{inconsistent}.
	 
	Let $r_1,r_2$ be two restrictions that restrict two disjoint sets of positions $I_1,I_2$. 
	We denote with $r_1\hat{}r_2$ the restriction $r$ that restricts positions $I_1\sqcup I_2$ such that $r$ restricts the positions in $I_1$ to the same values as $r_1$, and positions in $I_2$ to the same values as $r_2$. We say that $r$ is the concatenation of $r_1$ and $r_2$.
	
	Let $r_1,r_2$ be two restrictions that restrict two nested sets of positions $I_1\subseteq I_2$.
	If $r_2$ has the same values as $r_1$ on positions in $I_1$,
	we say that $r_2$ is an extension of $r_1$, and we write
	$r_1\prefixof r_2$.	
\end{definition}
	Note that every sequence is consistent with the empty restriction,
	and that the set of sequences consistent with a finite restriction is clopen. 
	
	If two restrictions are inconsistent, then the sets of sequences consistent with the restrictions are disjoint.

	A set of sequences consistent with concatenation of two restrictions is the intersection of sequences consistent with one restriction with sequences consistent with the other restriction. 

	If restriction $r'$ is an extension of $r$, then the set of sequences consistent with $r'$ is a subset of sequences consistent with $r$.

\section{Definitions}

\begin{definition}
	A \textit{partition refinement} of a set $\Omega$
	is a partial function $S$
	that maps $t,x\in\bbN\times\strings$ to nonempty subsets of $\Omega$.
	Instead of writing $S(t,x)$ we'll write $S^t(x)$.
	When defined, we call the set $S^t(x)$ a \textit{part} of the partition refinement $S$ \textit{at time $t$, with the coordinate $x$}.
	It has the following properties:
	\begin{itemize}
		\item 
		The empty string is mapped to the whole set $\Omega$, that is, $S^0(\Lambda)=\Omega$
		
		\item
		If $S$ is defined on $t,x$, then for all $t'>t$ it's value remains the same, that is, $S^{t'}(x)=S^t(x)$.

		\item 
		If $S$ is defined on $t,x0$, then it is also defined on $t-1,x$ and on $t,x1$. Furthermore, $\{S^t(x0),S^t(x1)\}$ is a partition of $S^t(x)$.
	\end{itemize}
	If for some $t,x$, $S^t(x)$ is defined and
	$S^{t}(x0), S^{t}(x1)$ are undefined, we'll call both the coordinate $x$, and the part $S^t(x)$ \textit{terminal at $t$} (w.r.t. the partition refinement $S$). 
	If $S^{t}(x)$ is a terminal part at $t-1$, but not at $t$,
	we say that this part
	\textit{is split at $t$ into two parts, $S^{t}(x0), S^{t}(x1)$} .\\
	For a part $S^t(x)$, we say that the part \textit{was split $|x|$ many times}.
	
		Let $S,S'$ be partition refinements	of the same set $\Omega$, such that 
		for all coordinates for which $S$ is still undefined at time $t$,
		$S'$ remains undefined at all times, and for all the remaining coordinates
		$S'$ is the same as $S$.
		We'll say that $S'$ is \textit{$S$ up to time $t$}.
		
\end{definition}
\begin{definition}
	A \textit{mass function} $\mu$ is a map from strings to non-negative reals
	with the property that for any string $x$, $\mu(x)=\mu(x0)+\mu(x1)$.
\end{definition}
\newcommand{\BS}{\text{BS}}
\begin{definition}[Betting Strategy]
	Let $S$ be a partition refinement of the set of infinite binary sequences whose range are clopen subsets of $\bseqs$.
	Let $\mu$ be a mass function.
	A pair $\BS=(S,\mu)$ is called a \textit{betting strategy}.
	A betting strategy is \textit{computable}, if $S$ is computable and, 
	for any coordinate $x$, 
	if there is some $t$ such that $S^t(x)$ is defined,
	then $\mu(x)$ halts.
	
	For the betting strategy $\BS$ and a coordinate $x$ if there is some $t$ such that the part $S^t(x)$ is defined, the \textit{capital} of the coordinate $x$ is defined as the value $c(x)=\frac{\mu(x)}{\lambda(S^t(x))}$.
	
	The \textit{maximum capital}, $\hat{c}$,  of a coordinate $x$ is the maximum of capital over all coordinates that prefix $x$, that is,
	$\hat{c}(x)=\max_{x'\prefixeqof x}c(x')$.
	
	For a sequence $\sigma$ and a coordinate $x$ such that $x$ is terminal at $t$ and $\sigma\in S^t(x)$,
	we'll say that the betting strategy up to time $t$ \textit{achieved capital}
	$\hat{c}(x)$ \text{(when betting) on} $\sigma$.
	
	The limit of achieved capital up to time $t$ on a sequence $\sigma$, when $t$ goes to infinity,
	is called 
	\textit{the achieved capital of the betting strategy on $\sigma$}. 
	If the achieved capital is unlimited, the betting strategy \textit{wins on $\sigma$}. 
\end{definition}

\begin{definition}
	A \textit{Kolmogorov-Loveland (KL) partition refinement} is a partition refinement of sequences whose range are sets of sequences that are consistent with some restriction.
\end{definition}

\begin{definition}
	A \textit{Kolmogorov-Loveland betting strategy (KLBS)} is a betting strategy that has a KL partition refinement.
	A sequence is \textit{Kolmogorov-Loveland random (KLR)} if no computable KLBS wins on it.
\end{definition}

\begin{remark}
	This definition of KLBS is equivalent to the more standard one, given in the introduction.
	
	For a partition refinement of sequences $S$ and the coordinate $\Lambda$, $S^0(\Lambda)$ is the set consistent with the empty restriction. 
	Suppose at some $t>0$,  the part $v=S^0(\Lambda)$ is split into parts
	$v_0=S^t(0),v_1=S^t(1)$.
	If $S$ is a KL partition refinement,
	then there are two restrictions $r_0,r_1$ such that
	$v_0=[r_0]$ and $v_1=[r_1]$.
	
	Both $r_0$ and $r_1$ can restrict only one position,
	as otherwise $\lambda(v_0)+\lambda(v_1)<\lambda(v)$,
	and $\{v_0,v_1\}$ would not be a partition of $v$,
	and for the same reason $r_0,r_1$ must restrict the same position
	to different values, as otherwise the sets $v_0,v_1$ intersect.
	
	Thus, from the KL partition refinement we can obtain the function that chooses the next position to bet on, and vice-versa.
\end{remark}

\begin{definition}
	Let $f$ be a computable function from $\bbN\times\bbN$ to basic sets.
	Let $n\in\bbN$ and let $f_n$ be the union over all $k$ of basic sets $f(n,k)$.
	
	If for all $n$, the size of the (open) set $f_n$ is less than $2^{-n}$,
	and $f_{n+1}\subseteq f_n$, we call $f_n$ the $n$-th level of the ML test.
	If a sequence is in the $n$-th level of the ML test, we say it fails the test at the $n$-th level, and if it fails at all levels we say it fails the test.
	The intersection $\bigcap_n f_n$ is called an effective nullset, this is precisely the set of sequences that fail the ML-test.
	
	A sequence is non-\textit{Martin-L\"of random (MLR)} 
	if it it fails some Martin-L\"of test 
	(equivalently, the sequence is in some effective nullset).
\end{definition}

\section{Finite game}
In section \ref{sec_Basic_game}, from a pair of betting strategies, we construct an open set that has small size and contains a sequence
on which neither betting strategy wins, if the betting strategies have splits bounded in a certain way.

The basic sets that get chosen into this open set depend on the bets of the strategy pair.
We can view this as a game between a player that constructs an open set 
by choosing a small amount of basic sets against
a player that constructs a pair of betting strategies that must observe the bound on splits.
We'll show that the first player has a way of winning in this game,
that is, some of the chosen basic sets will contain sequences on which neither betting strategy achieves capital larger than some fixed treshold.

For the purpose of the analisys of this construction
we will look at the projection of this game on the set of sequences
consistent with a restriction that leaves only finitely many positions unrestricted.

\begin{definition}
	Let $S$ be a partition refinement of sequences.
	We say that a restriction $r$ is \textit{elementary at $t$ w.r.t. $S$} 
	if there is terminal part $v$ (at $t$ w.r.t. $S$) that contains $[r]$.
\end{definition}
\begin{remark}
	A restriction that restricts all positions is elementary  w.r.t. any partition refinement of sequences at all times because the set of sequences consistent with such restriction contains a single sequence.
\end{remark}

\begin{definition}
	Let $S$ be a partition refinement of sequences, and
	let 
	$I,J$ be two disjoint finite subsets of positions.
	
	We'll say that a restriction $r\in\{0,1\}^J$ is \textit{$I$-granular at $t$ w.r.t. $S$} if 
	for all $s\in\{0,1\}^I$ the restriction $s\hat{}r$ is elementary (at $t$ w.r.t. $S$).
\end{definition}

\begin{definition}	
	Let $r$ be an $I$-granular restriction at $t$ w.r.t. $S$
	and 
	let $v$ be a part of a partition refinement $S$ that is split at $t$
	into parts $v_0,v_1$.
	We'll say that \textit{$v$ is split on $r$ at $t$ w.r.t. $S$} when
	both $v_0$ and $v_1$ intersect $[r]$.
	
	Let $x$ be a coordinate of a part $v$ that is defined at $t$ and
	$N$ the number of coordinates $x'\prefixof x$ such that the part
	with the coordinate $x'$ is split on $r$ (at some $t'<t$).
	We say that the part $v$ \textit{was split on $r$ $N$ times}.
	
\end{definition}

\begin{definition}
	Let $I$ be a finite subset of $\bbN$.
	We denote with $R_I$ the set of restrictions that restrict all of the positions except the ones in $I$, that is $R_I=\{0,1\}^{\bbN\setminus I}$.
\end{definition}

\begin{remark}
	For any subset of positions $I$ a restriction that restricts
	all positions except the ones in $I$ is $I$-granular at all times w.r.t. any partition refinement of sequences.
\end{remark}

\begin{definition}\label{def_I-split_on}
	Let $I$ be some finite set of positions, $s$ a restriction that restricts $I$, and $\rho$ a restriction that restricts all positions
	except the ones in $I$.
	Denote the (only) string consistent with the restriction
	$s\hat{}\rho$ with $\sigma$.
	Let $N$ be the maximum,
	over parts in a partition refinement of a betting strategy
	that contain $\sigma$,
	of the number of times a part was split on $\rho$.
	
	We'll say that $\sigma$
	was \textit{$I$-split on $N$ many times 
		by the betting strategy}.
\end{definition}

\begin{definition}
	Suppose that the partition refinement of the betting strategy is such that for every $\rho$ with finitely many unrestricted positions,
	whenever a part $v$ is split on $\rho$ into parts $v_0,v_1$,
	both $v_0$ and $v_1$ contain the same number of elements from $[\rho]$.
	We'll call such betting strategies \textit{half-splitting}.
\end{definition}

\begin{remark}
	The Kolmogorov-Loveland betting strategies are half-splitting.
	Moreover, when part $v$ is split into parts $v_0,v_1$ if $v$ is split on $\rho\in R_I$,
	it is also split on every $\rho'\in R_I$ that has non-empty intersection $[\rho']\cap v$.
	Furthermore, there is a position $p\in I$ such that $v_0$ contains the sequences in $v$ that have $0$ at position $p$ and $v_1$ the sequences that have $1$
	at position $p$.
\end{remark}

\begin{definition}\label{def_proj}
	Let $S$ be a partition refinement of sequences with clopen parts,
	$I$ a finite set of positions,
	and $r$ a restriction that is $I$-granular w.r.t. $S$ at all times.
	
	We'll recursively define a partition refinement $T$ of restrictions
	in $\{0,1\}^{I}$ from
	the partition refinement $S$.
	
	Set $T^0(\Lambda)=\{0,1\}^{I}$.
	Let $u,v$ be terminal parts at $t-1$ w.r.t. $T,S$ such that $\bigsqcup_{s\in u}[s\hat{}r]=v\cap[r]$ (\textit{corresponding parts}).
	If $v$ does not split at $t$, both $v$ and $u$ remain terminal parts at $t$ w.r.t. $S$ and $T$.
	
	If $v$ does split at $t$ (into parts $v_0,v_1$), but does not split on $r$,
	part $u$ remains terminal at $t$ w.r.t. $T$ and corresponds to
	the part in $\{v_0,v_1\}$ that contains $v\cap[r]$.
	
	If $v$ splits on $r$ at $t$ (w.r.t. $S$),
	then $u$ splits at $t$ (w.r.t. $T$) into parts $u_0, u_1$ such that
	$u_0$ corresponds to $v_0$ and $u_1$ to $v_1$.
	
	We'll say that 
	$T$ is $S$ \textit{projected on $r$}
\end{definition}

\begin{definition}
	An \textit{evaluation function} 
	is a function
	$\nu^t(x)$ that maps $t,x\in\bbN,\strings$ to non-negative reals.
	It is non-decreasing both in $t$ and continuations of $x$.
	That is, for any $t\leq t',\;x \prefixeqof x'$, $\nu^t(x)\leq \nu^{t'}(x')$.
\end{definition}

\begin{definition}
	A \textit{partition evaluation} is a pair of
	a partition refinement of a finite set and an evaluation function.
\end{definition}

\begin{definition}\label{def_BS_restricted_to_rho}
	Let $\BS=(S,\mu)$ be a betting strategy, $I$ a finite set of positions and $r$ a restriction that that is $I$-granular w.r.t. $S$ at all times.
	Let $T$ be $S$ projected on $r$.
	Let $x$ be a terminal coordinate at $t$ w.r.t. $T$ and
	$x'$ a terminal coordinate at $t$ w.r.t. $S$ such that
	parts $T^t(x)$ and $S^t(x')$ correspond.
	Let $\nu^t(x)=\hat{c}(x')$.
	We'll say that the pair $(T,\nu)$ is \textit{$\BS$ projected on $r$}. 
\end{definition}
\begin{lemma}
		Let $I$ be a finite subset of $\bbN$, $\rho$ a restriction in $R_I$.
		A betting strategy projected on $\rho$ is a partition evaluation.
\end{lemma}
\begin{proof}
	Let $(T,\nu)$ be some betting strategy projected on $\rho$.
	Partition refinement $T$ is a partition refinement of the set $\{0,1\}^I$.
	This set is finite since it contains $2^{|I|}$ restrictions.
	The evaluation function $\nu$ is non-decreasing both in $t$ and continuations of $x$ because the maximal capital of the betting strategy also has those properties.
\end{proof}

In section \ref{sec_Basic_game}, from a pair of betting strategies $\BS_A,\BS_B$, and additional parameters: 
a finite set of positions $I$, 
an upper bound on the number of splits $N$, 
and a lower bound on the number of elements $\phi$,
we will construct
a sequence of finite restrictions, $C$.
(And with it the open set of (infinite binary) sequences consistent with those restrictions.)

For a $\rho\in R_I$, $C$ will have a subsequence of the form
$C'=s_1\hat{}r_1,\dots,s_z\hat{}r_z$,
where $s_i$ restricts $I$, and 
$r_1\prefixeqof\dots\prefixeqof r_z\prefixof\rho$.
Let partition refinements $P_A=(T_A,\nu_A)$, $P_B=(T_B,\nu_B)$ be 
the betting strategies $\BS_A,\BS_B$ projected on $\rho$.
The sequence $C'$ will have the property
that
there are 
times $t_0=0,t_1,\dots,t_z$ such that:
\begin{itemize}
	\item
	At $t_{i-1}$, the restriction $s_i$ is contained in a pair of terminal parts $a,b$ (w.r.t. $T_A,T_B$) such that both parts have split less than $N$ many times,
	have more than $\phi$ many elements, and their evaluations are less than $2$. We say that $s_i$ is $(N,\phi)$-good at $t_{i-1}$ w.r.t. $P_A,P_B$.
	\item
	Let $t_i$ be the smallest $t>t_{i-1}$ such that either part $a$
	or part $b$ split at $t$, or the evaluation of either parts becomes large (larger than $2$) at $t$. If there is such $t$, we say that $s_i$ becomes stale at $t_i$  w.r.t. $P_A,P_B$. If there is no such $t$, $t_i$ remains undefined and we
	say that $s_i$ is forever-fresh after $t_{i-1}$, and in this case,
	$s_i\hat{}r_i$ is the last restriction in $C'$.
	\item
	If $s_i$ becomes stale at $t_i$, and there are still some
	$(N,\phi)$-good restrictions in $\{0,1\}^I$ at $t_i$ w.r.t. $P_A,P_B$, one of them is the sequence $s_{i+1}$. If there are no 	$(N,\phi)$-good restrictions, then $s_i\hat{}r_i$ is the last restriction in $C'$.
\end{itemize}

$C'$ cannot have too many elements since the restrictions $s_i$ are chosen to be in the intersection of two terminal parts of $T_A,T_B$ that are large (have more than $\phi$ many elements).

Because of this, when $s_i$ becomes stale because an evaluation of a part becomes large, then also $\phi$ other restrictions are contained
in this part with large evaluation. This can happen at most $2^{|I|}/\phi$ many times per partition evaluation, 
because then, at some $t$, all of the terminal parts have large evaluation,
and every restriction in $\{0,1\}^I$ is $(N,\phi)$-bad.

Similarly, when $s_i$ becomes stale because a part was split,
then also $\phi$ other sequences are contained
in parts whose number of splits was incremented by $1$.
This can happen at most $N2^{|I|}/\phi$ many times per partition evaluation.
Therefore, for a pair of general betting strategies,
$C'$ can have at most $2(N+1)2^{|I|}/\phi$ elements.

For the half-splitting betting strategies,
we can improve the upper bound on the number of elements in $C'$.
The size of a part in 
the projection of a half-splitting betting strategy on $\rho\in R_I$,
and the number of times the part was split are related, namely, if a part 
was split $n$ many times,
its size is $2^{|I|-n}$.
Let $M$ be the smaller of $N,|I|-\log\phi$. 
Then $M$ is the upper bound on the
number of times the part was split, 
before all of the restrictions in the part become $(N,\phi)$-bad.
The number of parts that were split less than $M$ times is at most
$2^{M+1}$, so at most $2^{|I|-\log\phi+1}=2\cdot2^{|I|}/\phi$ restrictions in $C'$
become stale because a part in the projection of the half-splitting betting strategy was split.
Therefore, for a pair of half-splitting betting strategies,
$C'$ can have at most $2(2+1)2^{|I|}/\phi$ elements.

\begin{definition}\label{def_N,phi-good}
	Let $S$ be a partition refinement of a finite set $\Omega$.
	Let $P=(S,\nu)$ be a partition evaluation. 
	
	We'll say that an element of $\Omega$ is \textit{$(N,\phi)$-good} at $t$ if a part that is terminal at $t$ w.r.t. $S$ that contains the element
	was split less than $N$ times, and has more than $\phi$ many elements, and the evaluation of the part's coordinate is less than $2$.
	If the part is not $(N,\phi)$-good at $t$ it is \textit{$(N,\phi)$-bad} at $t$.
\end{definition}


\begin{definition}\label{def_stale_el}
	Let $S$ be a partition refinement of a finite set $\Omega$.
	Let $P=(S,\nu)$ be a partition evaluation. 
	
	An element of $\Omega$ becomes stale after $t$ at $t'$ w.r.t. $P$
	if $t'$ is the smallest number larger than $t$ such that 
	a part of $S$ that contains the element is split at $t'$,
	or,
	the evaluation of a part that contains the element becomes larger than $2$ at $t'$.
\end{definition}

\begin{definition}\label{def_good_bad_stale_wrt_PE_pair}
	Let $S_A,S_B$ be partition refinements of a finite set $\Omega$,
	and $P_A=(S_A,v_A), P_B=(S_B,v_B)$ a pair of partition evaluations.
	
	An element of $\Omega$ is $(N,\phi)$-good at $t$ w.r.t. $P_A,P_B$ if 
	it is  $(N,\phi)$-good at $t$ w.r.t. both $P_A$ and $P_B$.
	It is $(N,\phi)$-bad w.r.t. $P_A,P_B$ if it is $(N,\phi)$-bad w.r.t. $P_A$ or w.r.t. $P_B$.
	
%
	
	An element of $\Omega$ becomes stale after $t$ at $t'$ w.r.t. $P_A,P_B$
	if it stale after $t$ at $t'$ w.r.t. $P_A$ or w.r.t. $P_B$.
	If the element never becomes stale after $t$, we'll say that the element is \textit{forever-fresh} after $t$.
\\

	Let $C=\{s_1,\dots,s_k\}$ be a sequence of elements from $\Omega$.
	
	Let $t_0=0$ and for every $i\in[1,k]$ let $t_i$ be such that $s_i$ becomes stale after $t_{i-1}$ at $t_i$ (w.r.t. $P_A$,$P_B$) (if $s_i$ is forever-fresh after $t_{i-1}$, $t_j$ is undefined for all $j\geq i$ ).
	We say that $C$ is an \textit{$(N,\phi)$-sequence w.r.t. $P_A,P_B$} if for all
	$i\in [1,k]$, $t_{i-1}$ is defined, and  
	$s_i$ is $(N,\phi)$-good at $t_{i-1}$.
\end{definition}
\begin{proposition}\label{Ccount}
	Let $S_A,S_B$ be partition refinements of a set $\Omega$ with $2^\ell$ many elements,
	and $P_A=(S_A,\nu_A), P_B=(S_B,\nu_B)$ a pair of partition evaluations.
Let $C$ be an $(N,\phi)$-sequence w.r.t. $P_A,P_B$.
The sequence $C$ has less than $2\frac{2^\ell}{\phi}(N+1)$ elements.

Furthermore, if $S_A,S_B$ are such that whenever a part is split, it is split into two parts with equal number of elements,
then $|C|\leq 6\frac{2^\ell}{\phi}$
\end{proposition}
\begin{proof}
Let $a_i, b_i$ be the parts with coordinates $x_i,y_i$ that are terminal at $t_{i-1}$ w.r.t. $S_A,S_B$ and
contain $s_i$, the $i$-th element of $C$.

If $s_i$ becomes stale at $t_i$ then
\begin{enumerate}[label=(\alph*)]
	\item\label{staleA}
	\begin{enumerate}[label=(\arabic*)]
		\item\label{evalAti} the evaluation $\nu_A$ of $x_i$ becomes larger than $2$, or
		\item\label{splitAti} $a_i$ is split at $t_i$, or 
	\end{enumerate} 
	\item\label{staleB} 
	\begin{enumerate}[label=(\arabic*)]			
		\item\label{evalBti} the evaluation $\nu_B$ of $y_i$ becomes larger than $2$, or
		\item\label{splitBti} $b_i$ is split at $t_i$. 
	\end{enumerate} 
\end{enumerate} 
Since $s_i$ is $(N,\phi)$-good at $t_{i-1}$, 
both $a_i$ and $b_i$ contain more than $\phi$ many elements and both of them have split less than $N$ times.

In case \ref{staleA}\ref{evalAti}, at least $\phi$ many elements
will be in a part with evaluation larger than $2$. Since evaluation function is nondecreasing both in $t$ and in further splits, for all $t\geq t_i$ these elements remain $(N,\phi)$-bad and therefore cannot show up later in the sequence $C$. Thus \ref{staleA}\ref{evalAti} can happen for at most
$\frac{2^\ell}{\phi}$ many elements in $C$ before all of the elements in $\Omega$ become $(N,\phi)$-bad and $s_i$ is the last element in $C$.

In case \ref{staleA}\ref{splitAti}, for at least $\phi$ many elements
the number of times the partition $S_A$ has split increases by 1.
Once a part has split more than $N$ times, the elements in that part become $(N,\phi)$-bad and cannot show up later in the sequence $C$.
Therefore \ref{staleA}\ref{splitAti} can happen for at most
$\frac{2^\ell}{\phi}N$ many elements in $C$, in which case $s_i$ is the last element in $C$.

Summing \ref{staleA}\ref{evalAti} and \ref{staleA}\ref{splitAti}
we have that \ref{staleA} happens for at most $\frac{2^\ell}{\phi}(N+1)$
elements in $C$, and, by the same reasoning, this is also true for \ref{staleB}. Then, in total, $C$ has less than   $2\frac{2^\ell}{\phi}(N+1)$ many elements.

When $S_A,S_B$ are half-splitting, the analisys of the case \ref{staleA}\ref{evalAti} remains the same,
however, we can improve the bound on the size of the sequence $C$ in the case \ref{staleA}\ref{splitAti}.

If $S_A$ is such that whenever a part is split, 
it is split into two parts of equal size, then
$|a_i|=2^{\ell-|x_i|}$.
Since the element $s_i\in a_i$ is $(N,\phi)$-good at $t_{i-1}$ (w.r.t. $P_A$),
we have that $x_i$ is shorter than $N$ (otherwise the part $a_i$ has split more than $N$ times), but also
it must be shorter than $\ell-\log\phi$ as otherwise the part 
$a_i$ has size smaller than $\phi$.

Let $M=\min(N,\ell-\log\phi)$.
If $s_i$ is $(N,\phi)$-good at $t_{i-1}$
then the length of $x_i$ is at most $M$.
If $|x_i|=M$ and $a_i$ is split at $t_i$, it is split into two parts $a',a''$ with coordinates $x_i1,x_i0$,
and since $|x_i1|,|x_i0|$ are both equal to $M+1$, the elements in both parts
$a',a''$ are $(N,\phi)$-bad at all $t\geq t_i$.
But the number of coordinates shorter than $M$ is less than $2^{M+1}$,
and therefore the case \ref{staleA}\ref{splitAti} can happen at most
$2^{M+1}$ times before all of the elements in $\Omega$ become $(N,\phi)$-bad and $s_i$ is the last element in the sequence $C$.

Since $M\leq \ell-\log\phi$, the case \ref{staleA}\ref{splitAti} happens for at most
$2^{\ell-\log\phi+1}=2\frac{2^\ell}{\phi}$ many elements in $C$,
and the case \ref{staleA}\ref{evalAti} for at most
$\frac{2^\ell}{\phi}$. Together, \ref{staleA} happens for at most
$3\frac{2^\ell}{\phi}$ elements, and the same is true for \ref{staleB},
which brings us to a total of at most $6\frac{2^\ell}{\phi}$ elements in $C$.
\end{proof}

We show that, for projections of betting strategies on $\rho\in R_I$, if the bound $N$ on the number of splits is smaller than $|I|-\log\phi$ by $h$, only a small ($2^{-h}$) fraction of sequences in $[\rho]$ can be contained in parts that are smaller than $\phi$ but have split less than $N$ many times.
Of course,
for half-splitting betting strategies, this fraction is $0$.

\begin{lemma}\label{slim_approx_splitted}
	For any partition refinement of a finite set $\Omega$ with $2^\ell$ many elements, and any $t$,
	the total number of strings that are contained in terminal parts that have
	strictly less than $\phi$ many elements and have split at most $\ell-\log\phi-h$ times is less than $2^{\ell-h}$.
\end{lemma}
\begin{proof}
	Let $k=\ell-\log\phi-h$
	The number of terminal coordinates for a partition $S^t$ that are shorter than $k$ is at most $2^k$. Even if all of these coordinates are mapped by $S^t$
	to parts with (strictly) less than $\phi$ many elements, the total  number of elements in such parts is less than $2^k\phi=2^{\ell-h}$. 
\end{proof}
\begin{lemma}\label{size_of_slim,subN-splits}
	Let $N=\ell-\log\phi-h$
	The number of $(N,\phi)$-bad elements at $t$ w.r.t. $P_A,P_B$ that are contained in terminal parts with evaluation below $2$
	and have split less than $N$ times,
	is at most
	$
	2^{\ell-h}
	$.
\end{lemma}
\begin{proof}
	By lemma \ref{slim_approx_splitted}. 
\end{proof}

\section{Basic game}\label{sec_Basic_game}

For a given pair of betting strategies, an interval of positions $I$, 
and
a finite restriction $z$ that restricts positions
not in $I$ 
, we will construct a set of finite restrictions $C$, the chosen restrictions.
The chosen restrictions will be extensions of $z$, and the open set
$\bigcup_{c\in C}[c]$ will have a fraction of the size of $[z]$.

If $[z]$ has a large ($>0$) subset of sequences
on which the betting strategies achieve only small ($\leq2$) capital,
and under a certain condition on the splits of the betting strategies,
some of the chosen restrictions will also have a large subset of sequences on which the betting strategies achieve small capital.

In section \ref{sec_ML_test} we'll iteratively use sets of chosen restrictions to construct levels of a ML-test that contains a 
sequence on which the betting strategy pair achieves only small capital, 
showing that there is a non-MLR sequence on which the betting strategy pair doesn't win.
\\

The restrictions are chosen in such a way that for any $\rho\in R_I$ that extends  $z$, the subset of $C$ consistent with $\rho$ will be of the form
$\{s_1\hat{}r_1,\dots,s_z\hat{}r_z\}$
where $\rho$ is an extension of all $r_1,\dots, r_z$,
and $s_1,\dots,s_z$ is an $(N,\phi)$-sequence w.r.t. the pair of betting strategies projected on $\rho$.

\begin{lemma}\label{proj_on_r<r'_eq}
	Let $\BS=(S,\mu)$ be a betting strategy and
	$r$ an $I$-granular restriction at $t$ w.r.t. $S$.
	Any $r'$ that extends $r$ and leaves positions in $I$ unrestricted
	is also $I$-granular at $t$ w.r.t. $S$.
	Furthermore,
	$\BS$ up to time $t$ projected on $r$ is equal to
	$\BS$ up to time $t$ projected on $r'$.
\end{lemma}
\begin{proof}
	Since $r$ is $I$-granular (at $t$ w.r.t. the partition refinement of $\BS$),
	for all $s\in\{0,1\}^I$, $s\hat{}r$ is elementary,
	and if $s\hat{}r$ is elementary so is $s\hat{}r'$.
	Therefore, $r'$ is also $I$-granular.
	
	Furthermore, for any part $v$, defined by the partition refinement
	of the $\BS$ up to time $t$, if the intersection $v\cap[r]$ is non-empty, 
	then, by definition \ref{def_proj},
	$v$ has corresponding parts $u,u'$ in projections of the partition refinement on $r,r'$, both having the same coordinate and defined at the same time as part $v$.
	By definition \ref{def_BS_restricted_to_rho},
	the evaluation functions of projections of the $\BS$ up to time $t$
	on $r,r'$ are the same as well, for all coordinates and times.
	But then, the projections of $\BS$ up to time $t$ on $r,r'$ are the same.
\end{proof}

\begin{definition}\label{def_good_bad_stale_wrt_BS_pair}
	Let $I$ be some finite set of positions and
	let $\BS_A=(S_A,\mu_A)$, $\BS_B=(S_B,\mu_B)$ be a pair of betting strategies.
	Let
	$s$ be a restriction that restricts $I$ and 
	$r$ an $I$-granular restriction at $t$ w.r.t. $S_A,S_B$. 
	
	The restriction $s\hat{}r$ is \textit{$(I,(N,\phi))$-good at $t$ w.r.t. $\BS_A,\BS_B$} iff $s$ is $(N,\phi)$-good at $t$ w.r.t. both projections on $r$ of $\BS_A$ and $\BS_B$ up to time $t$. The restriction $s\hat{}r$ is $(I,(N,\phi))$-\textit{bad} (at $t$ w.r.t. $\BS_A,\BS_B$) if it is not $(I,(N,\phi))$-good.
	
	We'll say that $s\hat{}r$ \textit{becomes $I$-stale after $t$ at $t'$ w.r.t. $\BS_A,\BS_B$} iff $s$ becomes stale after $t$ at $t'$ w.r.t. both projections on $r$ of $\BS_A$ and $\BS_B$ up to time $t'$.
\end{definition}

To make the construction of the set of chosen sequences exact,
for a given betting strategy pair, at time $t$,  
we use a uniquely defined set of $I$-granular restrictions, called the
common set of  $I$-granular restrictions at time $t$.
\begin{lemma}\label{maxpos}
	Let $S$ be a partition refinement whose parts are clopen.
	For every $t$ there is a set of positions $K$ such that all of the restrictions in $\{0,1\}^K$ are elementary at $t$ w.r.t. $S$.
\end{lemma}
\begin{proof}
	A basic set is a set of sequences that extend some string.
	A clopen set consists of finitely many basic sets.
	There are finitely many parts terminal at $t$ w.r.t. $S$. 
	These are clopen sets and therefore there is a minimal length of a string $\ell$
	such that each terminal part can be represented as a union of basic sets that extend strings of length $\ell$.
	
	Let $K=[1,\ell]$, and $r$ a restriction in $\{0,1\}^K$.
	The set $[r]$ is a set of sequences that extend a string of length $\ell$, and is therefore a subset of some terminal part (at $t$ w.r.t. $S$).
	That is, $r$ is an elementary restriction (at $t$ w.r.t. $S$). 
\end{proof}
\begin{lemma}\label{minsetpos}
	Let $S$ be a partition refinement of sequences whose parts are clopen. For every $t$ there is a finite set of positions $K$ such that all of the restrictions in $\{0,1\}^K$ are elementary at $t$ w.r.t. $S$, and for every other set of positions $L$, if $\{0,1\}^{L}$ is a set of elementary restrictions, then $K\subseteq L$.
\end{lemma}
\begin{proof}
	By lemma \ref{maxpos} for all $S,t$ there is a finite set of positions such that the restrictions that restrict positions in that set are elementary at $t$.
	To prove the lemma, it is enough to show that for two finite sets of positions $K,L$ if both
	$\{0,1\}^K$ and $\{0,1\}^L$ are sets of restrictions elementary (at $t$, w.r.t. $S$)
	then $\{0,1\}^{K\cap L}$ is also a set of elementary restrictions.
	
	Let $o\in\{0,1\}^{K\cap L}$, $p\in\{0,1\}^{K\setminus L}$ and $r\in\{0,1\}^{L\setminus K}$.
	The restrictions $o\hat{}p$  and $o\hat{}r$ are both elementary since they are elements of $\{0,1\}^K$ and $\{0,1\}^L$, respectively.
	The restriction $o\hat{}p\hat{}r$ is also elementary, since 
	it extends the elementary restrictions $o\hat{}p$ and
	$o\hat{}r$. 
	For any $p\in\{0,1\}^{K\setminus L}$, $[o\hat{}p\hat{}r]$ is
	a subset of the terminal part that contains $[o\hat{}r]$,
	and for any $r\in\{0,1\}^{L\setminus K}$, $[o\hat{}p\hat{}r]$ is
	a subset of the terminal part that contains $[o\hat{}p]$.
	We conclude that there is a terminal part that contains $[o\hat{}p\hat{}r]$, for all $p,r\in\{0,1\}^{K\setminus L},\{0,1\}^{L\setminus K}$.
	But then, $[o]$ is a subset of that terminal part, that is, $o$ is an elementary restriction.
\end{proof}

\begin{definition}
	Let $S$ be a partition refinement of sequences whose parts are clopen. Let $K(S,t)$ denote the finite set of positions such that all of the restrictions in $\{0,1\}^{K(S,t)}$ are elementary at $t$ w.r.t. $S$, and for every other set of positions $L$, if $\{0,1\}^{L}$ is a set of elementary restrictions, then $K(S,t)\subseteq L$. By lemma \ref{minsetpos} $K(S,t)$ is defined for any $S,t$. We'll call $K(S,t)$  \textit{the positions inspected by $S$ up to time $t$}. 
\end{definition}

\begin{definition}	
	Let $A,B$ be a pair of partition refinements of sequences whose parts are clopen. 
	
	Let $K(A,B,t)$ denote the union of
	positions inspected by $A$ and by $B$ up to time $t$. We'll call $K(A,B,t)$ the set of \textit{positions inspected by $A,B$ up to time $t$}. 
	
	We'll say that a restriction $r\in\{0,1\}^J$ is \textit{$I$-granular at $t$ w.r.t. $A,B$} if it is $I$-granular at $t$ w.r.t. both $A$ and $B$.
	
	Let $K$ be the set of positions inspected by $A,B$ up to time $t$.
	We'll call the set of restrictions $\{0,1\}^{K\setminus I}$
	\textit{the common set of $I$-granular restrictions at $t$ w.r.t. $A,B$}.
\end{definition}

We'll now construct a choice function $C$ that, 
as time progresses, chooses finite restrictions into the chosen set.
We index the chosen restrictions with  finite sequences of numbers,
and for $m\in\bbN$, when $C(m)$ is defined, we'll say that $C(m)$ is 
the $m$-th chosen restriction. We denote the time when the $m$-th restriction was chosen with
$T(m)$.

The parameters for our construction are:
a pair of betting strategies $\BS_A,\BS_B$,
a finite set of positions $I$,
a restriction $z$ that doesn't restrict positions in $I$,
a lower bound on the number of splits $N$, and
an upper bound on the number of elements of a part in the finite game $\phi$.

Recall that $\_$ denotes the empty sequence of numbers.
The first chosen restriction $C(\_)=s_{\_}\hat{}r_{\_}$, is the concatenation of 
the restriction $s_{\_}$ that restricts all positions in $I$ to $0$
and the restriction $r_{\_}=z$.
The restriction $C(\_)$ is chosen at time $T(\_)=0$,
when partition refinements of both betting strategies have only one part defined that is the entire set of sequences.
It has the following property, that all of the chosen restrictions will have:
\begin{itemize}
	\item the restriction $C(\_)$ is a concatenation of two restrictions,
	$s_{\_}$ that restricts $I$, called the head of the chosen restriction and
	$r_{\_}$ called the tail.
	\item The tail
	is $I$-granular w.r.t. $\BS_A,\BS_B$ at time when the restriction was chosen.
	\item The head
	is $(N,\phi)$-good w.r.t. $\BS_A,\BS_B$ projected on the tail at the time when the restriction was chosen.
	\item For the tail, the head is uniquely defined,
	it is the lexicographically least $(N,\phi)$-good restriction w.r.t. $\BS_A,\BS_B$ projected on the tail at the time when the restriction was chosen.
\end{itemize}

In other words, the chosen restriction $C(m)=s_m\hat{}r_m$ 
is $(I,(N,\phi))$-good w.r.t. $\BS_A,\BS_B$ at the time when it was chosen, 
and the choice of head $s_m$ is uniquely defined for the tail $r_m$.

If at some $t>T(m)$, for some common $I$-granular restriction $r'$ that extends the tail $r_m$, the restriction $s_m\hat{}r'$ becomes $I$-stale w.r.t.
$\BS_A,\BS_B$, 
and there are still some $(N,\phi)$-good heads
w.r.t. $\BS_A,\BS_B$ projected on $r'$,
then
another restriction is chosen, with lexicographically least head $s'$ such that $s'\:\hat{}r'$ is $(I,(N,\phi))$-good w.r.t. $\BS_A,\BS_B$ at time $t$.
Suppose that $r'$ is the $n$-th such extension of $r_m$ so far.
Let $m'=m,n$ and 
set $C(m')=s'\:\hat{}r'$ and $T(m')=t$.

\begin{definition}
	Let  $\BS_A=(S_A,\mu_A),\BS_B=(S_B,\mu_B)$ be a pair of betting strategies.
	Let $I$ be a finite subset of positions 
	and $z$ a finite restriction such that none of the restricted positions are in $I$.
	Let $N,\phi$ be natural numbers.
	
	We'll recursively define a partial map $C$ from finite sequences of natural numbers to restrictions together with an auxilliary partial map $T$ from $\bbN^*$ to $\bbN$. For $m\in\bbN^*$, we'll say that
	$C(m)$ is \textit{the $m$-th chosen restriction which was chosen at $T(m)$ in the basic game on $I,z$ against $BS_A,BS_B$ with parameters $N,\phi$}.
	
	$C(\_)$ is the extension of restriction $z$ that restricts positions in $I$ to zero
	and $T(\_)=0$.
	
	Suppose that for some argument 
	$m\in\bbN^*$ the maps $C$ and $T$ are already defined and
	$C(m)=s\hat{}r$, where $s$ restricts positions in $I$, and $r$ is $I$-granular at $T(m)$ w.r.t. $S_A,S_B$.
	For $t\geq T(m)$, let 
	$F^t(m)$ denote the subset of 
	the common set of $I$-granular restrictions at $t$ w.r.t. $S_A,S_B$,
	such that $r'\in F^t(m)$ extends $r$ and $s\hat{}r'$ becomes $I$-stale after $T(m)$ at $t$.
	We'll say that the restrictions $F^t(m)$ 
	\textit{fail} at $t$ in the basic game on $I,z$ against $BS_A,BS_B$.
	
	We'll say that, at $t$, an $I$-granular restriction $r$ is \textit{viable}
	if there is a restriction $s\in\{0,1\}^I$ such that
	$s\hat{}r$ is $(I,(N,\phi))$-good, and otherwise 
	we say that $r$ is \textit{choiceless}. 
	Let $G^t(m)$ be the restrictions in $F^t(m)$ that are viable, and $H^t(m)$ the restrictions in $F^t(m)$ that are choiceless.
	
	For the failed restrictions that are still viable,
	we choose another extension on positions in $I$ that is $(I,(N,\phi))$-good.
	Let $i=\sum_{T(m)<t'<t}|G^{t'}|$.
	For $n\leq|G^t|$
	we define $C(m,(i+n))$ to be the restriction $s\hat{}r_n$ where $r_n$ is the $n$-th element of
	$G^t(m)$ and $s$ is the lexicographically least restriction in $\{0,1\}^I$
	such that $s\hat{}r_n$ is $(I,(N,\phi))$-good at $t$ w.r.t. $\BS_A,\BS_B$. We define $T(m,(i+n))=t$.
	
	We'll call the mapping $C$ \textit{the choice function on $z,I$ against the pair of betting strategies $\BS_A,\BS_B$}.
	The set of restrictions $\{C(m): m\in\dom C\}$ is called the set of restrictions \textit{chosen} by the choice function $C$.
	The (open) set of sequences $\bigcup_{m\in\dom C}[C(m)]$ is called the set of sequences \textit{chosen} by the choice function $C$.
\end{definition}

We have that for every $n_1,\dots,n_k\in\bbN^k$ for which the choice function is defined,
the restrictions $C(\_)=s_0\hat{}r_0,\dots, C(n_1,\dots,n_k)=s_k\hat{}r_k$
are such that the tails extend each other, $r_0\prefixeqof\dots\prefixeqof r_k$, and the heads are an $(N,\phi)$-sequence w.r.t. the projections of betting strategies on any $\rho$ that extends $r_k$ and restricts all of the positions except the ones in $I$.
By proposition \ref{Ccount}, there is a bound on the number of
elements in an $(N,\phi)$-sequence, $Q$.
So the choice function is defined only on sequences of numbers of
length at most $Q$.

Note that for any $m\in\bbN^*$, $n,n'\in\bbN,\;n\neq n'$ such that both
$C(m,n)$ and $C(m,n')$ are defined, the tails
of $C(m,n)$, $C(m,n')$ are inconsistent restrictions that extend the tail
of $C(m)$.
Therefore the sets of sequences consistent with the tails of the chosen restrictions
with indices of length $k$ are mutually disjoint subsets of the set of sequences consistent with $z$.
That is, denoting the tail of the $m$-th chosen restriction with $C_r(m)$,
$\bigsqcup_{m\in\bbN^k\cap\dom C}[C_r(m)]\subseteq [z]$.
Since for every chosen restriction the set of sequences consistent with the restriction has size, conditional on the tail, $2^{-|I|}$,
we have that $\sum_{m\in\bbN^k\cap\dom C}\lambda([C(m)])\leq 2^{-|I|}\lambda([z])$.
But then, the sum of sizes of the sets of sequences consistent with the chosen restrictions is less than
$Q2^{-|I|}\lambda([z])$.

\begin{lemma}\label{choice_depth}
	The choice function on $z,I$ with parameters $N,\phi$ against the pair of betting strategies is undefined on all sequences of numbers longer than 
	the maximal number of elements in the $(N,\phi)$-good sequence
	against a pair of partition evaluations of a set with $2^{|I|}$ elements.
\end{lemma}
\begin{proof}
	Let $m=n_1,\dots,n_z$ such that $C(m)$ is defined.
	Then $C(\_)=s_{0}\hat{}r_{0},C(n_1)=s_1\hat{}r_1,\dots,C(n_1,\dots,n_z)=s_z\hat{}r_z$ are restrictions such that
	$r_z$ is a restriction that extends all of the restrictions $r_0,\dots,r_{z-1}$, and
	$s_i\hat{}r_z$ is $(I,(N,\phi))$-good at $t_i=T(n_1,\dots,n_i)$,
	and for all $i<z$,
	$s_i\hat{}r_z$ becomes $I$-stale after $t_i$ at $t_{i+1}=T(n_1,\dots,n_i,n_{i+1})$,
	w.r.t. $\BS_A,\BS_B$.
	By lemma \ref{proj_on_r<r'_eq},
	for every $\rho\in R_I$ that extends $r_z$, we'll have that
	$s_0,\dots,s_z$ is an $(N,\phi)$-sequence w.r.t. to the pair of betting strategies projected on $\rho$.
	By proposition \ref{Ccount} it cannot have more than $Q$ elements, where 
	$Q=2^{|I|}\frac{6}{\phi}$ if the betting strategies are half-splitting,
	and $Q=2^{|I|}\frac{2(N+1)}{\phi}$ otherwise.
\end{proof}

\begin{lemma}\label{subrest_disjoint}
	Let $C$ denote the choice function on $z,I$  with parameters $N,\phi$ against a pair of betting strategies. 
	Let $m\in\bbN^*$ be such that $C(m)$ is defined and
	$C(m)=s\hat{}r$, where $s$ restricts positions in $I$.
	
	Let $F$ be the union over $t$ of the granular restrictions containing sequences from the $m$-th chosen restriction that fail at $t$ in the basic game on $I,z$ against the pair of betting strategies.
	
	The restrictions in $F$ are mutually inconsistent extensions of $r$.
\end{lemma}
\begin{proof}
	Let $J^t$ be the common set of $I$-granular restrictions at $t$ w.r.t. the pair of partition refinements of the betting strategies.
	We have that 
	$F=\bigcup_{t\geq T(m)}F^t(m)$.
	By definition, $F^t(m)$ is a subset of $J^t$, and the restrictions in $J^t$ are mutually inconsistent, since they restrict the same set of positions.
	
	All that is left to prove is that for any $t<t'$ and $r,r'\in F^t(m),F^{t'}(m)$ the restrictions $r,r'$ are inconsistent.
	By definition, $r\in F^t(m)$ implies that $s\hat{}r$ becomes $I$-stale after $T(m)$ at $t$ (w.r.t. the pair of betting strategies) .
	Suppose $r,r'$ are consistent.
	Since the set of inspected positions grows with $t$,
	$r$ restricts a subset of positions restricted by $r'$
	and therefore, if $r,r'$ are consistent, $r'$ is an extension of $r$.
	If $s\hat r$ becomes $I$-stale after $T(m)$ at $t$ then
	also $s\hat r'$ becomes $I$-stale after $T(m)$ at $t$.
	On the other hand, $r'\in F^{t'}(m)$ implies that
	$s\hat{}r'$ becomes $I$-stale after $T(m)$ at $t'>t$, a contradiction.
	Therefore $r,r'$ are inconsistent.
\end{proof}

\begin{proposition}\label{chosen_size}
	For any pair of betting strategies $\BS_A,\BS_B$, 
	a set of positions $I$
	and a restriction $z$ that restricts a finite set of positions disjoint from $I$,
	the sum of sizes of sets of sequences consistent with restrictions chosen by the choice function on $z,I$  with parameters $N,\phi$ against $\BS_A,\BS_B$ is less than
	$Q2^{-|I|}\lambda([z])$, where $Q$ is the bound on the size of an
	$(N,\phi)$-sequence against a pair of partition evaluations of a set with $2^{|I|}$ elements.
	
\end{proposition}
\begin{proof}
	Let $C$ denote the choice function on $z,I$  with parameters $N,\phi$ against $\BS_A,\BS_B$. Let $C_r(m)$ be the restriction that leaves the positions in $I$ unrestricted and is the same
	as restriction $C(m)$ on positions outside $I$.
	
	By definition
	$C_r(\_)=z$, 
	and by lemma \ref{subrest_disjoint},
	for any $m\in\bbN$,
	\\
	$\sum_{n\in\bbN\land m,n\in\dom C}
	\lambda([C_r(m,n)])
	\leq 
	\lambda([C_r(m)])$.
	Therefore, for any $k$,
	\\
	$\sum_{m\in\bbN^k\cap\dom C}\lambda([C_r(m)])\leq\lambda(z)$ implying
	$\sum_{m\in\bbN^k\cap\dom C}\lambda([C(m)])\leq 2^{-|I|}\lambda(z)$. 
	By lemma \ref{choice_depth}
	$C$ is defined only on sequences of natural numbers shorter than $Q$.
	We have:
	\\
	$
	\sum_{m\in\dom C}\lambda([C(m)])
	=
	\sum_{k\in\bbN}\sum_{m\in\bbN^k\cap\dom C}\lambda([C(m)])
	=
	\sum_{k\leq Q}\sum_{m\in\bbN^k\cap\dom C}\lambda([C(m)])
	\leq 
	Q2^{-|I|}\lambda(z)$
\end{proof}

For every $\rho$ that extends $z$ and restricts all of the positions except the ones in $I$,
since $(N,\phi)$-sequences are finite, there is the last chosen restriction
$s\hat{}r$,
whose tail $r$ is consistent with $\rho$.
The last restriction is such that either the head $s$ is forever-fresh after the time it was chosen
(w.r.t. $\BS_A,\BS_B$ projected on $\rho$), or, at some $t$,
for some tail $r'$, $r\prefixeqof r'\prefixof\rho$,
the restriction $s\hat{}r'$ becomes $I$-stale w.r.t. $\BS_A,\BS_B$
and there are no more $(I,(N,\phi))$-good restrictions with the tail $r'$ that can be chosen (choiceless $r'$).
We can divide the restrictions in $R_I$ that extend $z$ into two subsets, $H$,
the set of restrictions extending a choiceless restriction,  and $M$, 
the set of restrictions that extend the tail of some chosen restriction
whose head remains forever-fresh.

Note that the set of sequences consistent with $H$ is open, and with $M$ closed.

The set of sequences consistent with some $\rho\in H$ consists of:
\begin{itemize}
	\item 
	the set of sequences on
	which either betting strategy achieves high (larger than $2$) capital, 
	denoted with $W$
	\item 
	the set of sequences contained in a part of one of the partition refinements of $\BS_A,\BS_B$ that was split 
	more than $N$ times on $\rho$, 
	denoted with $L$
	\item 
	the set of sequences contained in a part $v$ of one of the partition refinements of $\BS_A,\BS_B$, such that $|v\cap[\rho]|<\phi$, denoted with $U$
\end{itemize}

The size of the set of sequences consistent with restrictions in $M$
is then larger than
$\lambda([z])-\lambda(W)-\lambda(L)-\lambda(U\setminus (L\cup W))$.
Assume that there is some $\theta,\epsilon$ such that
$\lambda(W|[z])\leq 1-\theta$ and $\lambda(L)\leq\epsilon\lambda([z])$.
Suppose $N\leq|I|-\log\phi-h$, then by lemma \ref{size_of_slim,subN-splits} $\lambda(U\setminus (L\cup W)|[z])\leq 2^{-h}$.
Then the size of the set of sequences consistent with restrictions in $M$
is larger than $\lambda([z])(\theta-\epsilon-2^{-h})$.

Let $M(m)$ denote the restrictions in $M$ consistent with the tail of the $m$-th chosen restriction and $\tilde{M}(m)$ the set of sequences consistent with restrictions in $M(m)$. For at least one $k$ smaller than
the maximum length of an $(N,\phi)$-sequence, $Q$,
we have
$\lambda(\bigsqcup_{m\in\bbN^k\cap\dom C}\tilde{M}(m))
\geq
\lambda([z])\frac{\theta-\epsilon-2^{-h}}{Q}
$.

But then for at least one $m\in\bbN^k$ we'll have that the size
of $\tilde{M}(m)$ conditional on the tail of the $m$-th chosen restriction
is $\theta'\geq \frac{\theta-\epsilon-2^{-h}}{Q}$.

This implies that for this $m$, the set of sequences consistent with 
$C(m)$ and some $\rho\in M(m)$ has size, conditional on $[C(m)]$, at least $\theta'$.
The betting strategies don't achieve high capital on any sequence in this set, and if $\theta'>0$ we have that
there is a large subset of sequences that have low capital consistent with the $m$-th chosen restriction. 

Recall the definition \ref{def_I-split_on}.
\begin{proposition}\label{chosen_low_cap}
	Let $\BS_A,\BS_B$ be a pair of betting strategies, 
	$I$ a finite set of positions,
	and $z$ a restriction that restricts a finite set of positions disjoint from $I$.
	Let $N,\phi,h$ be such that $N\leq |I|-\log\phi-h$.
	Let $\theta$ be the size, conditional on $[z]$, of the set of sequences on which the betting strategies achieve capital less than 2.
	Let $\epsilon$ be the size, conditional on $[z]$, of the set of
	sequences on which either betting strategy $I$-splits more then $N$ many times. 
	There is a restriction $c$ chosen by the choice function $C$ on $z,I$ with parameters $N,\phi$ against $\BS_A,\BS_B$ such that the size, conditional on $[c]$, of the set of sequences on which the betting strategies achieve capital less than 2 is at least
	$\frac{\theta-\epsilon-2^{-h}}{Q}$ where $Q=2\frac{2^{|I|}}{\phi}(N+1)$. 
\end{proposition}
\begin{proof}
	Let $C_s(m)$ denote the restriction that restricts positions in $I$, and $C_r(m)$ the restriction that restricts positions outside $I$
	such that $C(m)=C_s(m)\hat{}C_r(m)$.
	
	Let $F(m)$ denote the union over $t$ of 
	the granular restrictions containing sequences from the $m$-th chosen restriction that fail at $t$ in the basic game on $I,z$ against $\BS_A,\BS_B$.
	That is,
	$F(m)=\bigcup_{t}F^t(m)$.
	Let $\tilde{F}(m)$ be the set of sequences consistent with restrictions in $F(m)$. That is,
	$\tilde{F}(m)=\bigsqcup_{r\in F(m)}[r]$.
	
	Let $H(m)$ denote the choiceless restrictions in $F(m)$, and
	$\tilde{H}(m)$ the sequences consistent with them.
	That is,
	$H(m)=\bigcup_t H^t(m)$,
	$\tilde{H}(m)=\bigsqcup_{r\in H(m)}[r]$.
	
	Let $G(m)$ denote the viable restrictions in $F(m)$, and
	$\tilde{G}(m)$ the sequences consistent with them.
	That is,
	$G(m)=\bigcup_t G^t(m)$,
	$\tilde{G}(m)=\bigsqcup_{r\in G(m)}[r]$.

	Let
	$M(m)=[C_r(m)]\setminus\tilde{F}(m)$.
	That is,
	$M(m)$ is the set of sequences consistent with restrictions in $R_I$ that extend the tail of the $m$-th chosen restriction, $C_r(m)$, such that 
	for every $\rho\in M(m)$ the head of the $m$-th chosen restriction, $C_s(m)$,
	is forever-fresh after $T(m)$ w.r.t. $\BS_A,\BS_B$ projected on $\rho$.
	
	By lemma \ref{subrest_disjoint}, the restrictions in $F(m)$ are mutually inconsistent, and extend $C_r(m)$. 
	By definition, $F^t(m)=G^t(m)\sqcup H^t(m)$ and therefore
	$F(m)=G(m)\sqcup H(m)$. 
	We have that for all $m$ for which $C$ is defined,
	$[C_r(m)]=M(m)\sqcup\tilde{G}(m)\sqcup\tilde{H}(m)$.
	
	We have that $\bigcup_{n\in\bbN, (m,n)\in\dom C}C_r(m,n)=G(m)$, and we can write
	$[C_r(m)]=M(m)\sqcup\tilde{H}(m)\bigsqcup_{n\in\bbN,(m,n)\in\dom C}[C_r(m,n)]$.
	But then,
	$$[C_r(\_)]=
	(\bigsqcup_{0\leq k<Q}\bigsqcup_{m\in\bbN^k\cap\dom C}M(m)\sqcup\tilde{H}(m))
	\sqcup
	(\bigsqcup_{m\in\bbN^Q\cap\dom C}[C_r(m)])$$.
	By lemma \ref{choice_depth}, for large enough $Q$ and any $m\in\bbN^Q$, $C(m)$ is not defined. By definition, $C_r(\_)=z$, and we have
	\begin{equation}\label{CL2}
	[z]=\bigsqcup_{m\in\dom C}M(m)\sqcup\tilde{H}(m)
	\end{equation}

	Let $X^t$ be the subset of sequences
	consistent with a restriction that is $(I,(N,\phi))$-bad at $t$ w.r.t. $\BS_A,\BS_B$, 
	such that a sequence is in $X^t$
	only if it was $I$-split on less than $N$ many times by $\BS_A,\BS_B$ up to time $t$,
	and only if $\BS_A,\BS_B$, up to time $t$, achieve capital of less than $2$ on the sequence.
	
	Note that this implies that, for some $s\in\{0,1\}^I,\rho\in R_I$,
	the sequence consistent with the restriction $s\hat{}\rho$
	is in $X^t$, only if it
	is contained in a part of a partition refinement of one of the
	betting strategies up to time $t$, whose corresponding part in
	the projection on $\rho$ has less than $\phi$ many elements.
	
	Let $r\in H^t(m)$ and let $\rho$ be a restriction in $R_I$ that extends $r$. 
	By lemma \ref{proj_on_r<r'_eq}
	all of the sequences in $[\rho]$ are consistent
	with some $(I,(N,\phi))$-bad restriction w.r.t. $\BS_A,\BS_B$ up to time $t$.
	
	By lemma \ref{size_of_slim,subN-splits},
	there is at most $2^{|I|-h}$ sequences in $[\rho]$
	that are contained in intersections of parts $u,v$
	that split less than $N$ many times 
	w.r.t. $\BS_A,\BS_B$ up to time $t$ projected on $\rho$,
	and have evaluations $\nu_A(u),\nu_B(v)$ less than $2$.
	
	Summing (integrating) over $\rho\in R_I$ that extend $r$,
	we have that the size of $X^t$, conditional on $[r]$, is at most $2^{-h}$.
	This implies that the size, conditional on $[r]$,
	of the set of sequences contained in parts of
	the partitions refinements of the betting strategies that were split on $r$ more than $N$ many times or have maximal capital larger than $2$ is more than $1-2^{-h}$.
	
	For distinct $m,m'\in\bbN^*$, if $C(m),C(m')$ are defined,
	the sets $\tilde{H}(m),\tilde{H}(m')$ are disjoint subsets of $[z]$.
	The size, conditional on $[z]$, of the set of sequences on which the betting strategies achieve capital larger than $2$ is at most $1-\theta$, and we have
	$\sum_{m\in\dom C} \lambda(\tilde{H}(m))(1-2^{-h})\leq (1-\theta)\lambda([z])+\epsilon\lambda[z]$.
	This implies
	$$\sum_{m\in\dom C} \lambda(\tilde{H}(m))\leq ((1-\theta)+\epsilon+2^{-h})\lambda[z]$$

	Then from \eqref{CL2},
	$\sum_{m\in\dom C}\lambda(M(m))\geq (\theta-\epsilon-2^{-h})\lambda([z])$. 
	Let $Q$ be the maximal length of an $(N,\phi)$-sequence against a pair of partition evaluations of a set with $2^{|I|}$ elements.
	By lemma \ref{choice_depth}, for any $k\geq Q$,
	$C$ is not defined on any $m\in \bbN^k$.
	But then for at least one $k<Q$,
	$\sum_{m\in\bbN^k\cap\dom C}\lambda(M(m))
	\geq(\frac{\theta-\epsilon-2^{-h}}{Q})\lambda(z)$.
	Since
	$M(m)\subseteq [C_r(m)]$ and
	$\sum_{m\in\bbN^k\cap\dom C}\lambda([C_r(m)])\leq\lambda([z])$,
	for at least one $k,m\in\bbN^k$ we have
	$\lambda(M(m)|[C_r(m)])\geq \frac{\theta-\epsilon-2^{-h}}{Q}$.
	On the other hand,
	$\lambda(\bigsqcup_{\rho\in M(m)}[C_s(m)\hat{}\rho]|[C(m)])
	=
	\lambda(M(m)|[C_r(m)])
	$, 
	and since the $C_s(m)$ is forever-fresh after $T(m)$
	w.r.t. $\BS_A,\BS_B$ projected on $\rho$,
	the betting strategies achieve capital less than $2$ on 
	the sequence consistent with $C_s(m)\hat{}\rho$,
	and the result follows.
\end{proof}

\section{Constructing the Martin-L\"of-test}\label{sec_ML_test}
We'll now construct a ML-test for a given pair of computable betting strategies $\BS_A,\BS_B$.

We play the basic games on a sequence of disjoint sets of positions
$I_1,I_2,\dots$ paired up with parameters $(N_1,\phi_1),(N_2,\phi_2),\dots$, called game zones.
We'll call the pair $(I_i,(N_i,\phi_i))$ the $i$-th zone.

The basic game on $I_i,z$ against $\BS_A,\BS_B$
with parameters $N_i,\phi_i$ we call the basic game for $z$ on the $i$-th zone.

The $n$-th level of the ML-test will be the set of sequences consistent with
restrictions in the $n$-th level of the chosen restrictions (to be defined).
The $0$-th level of the chosen restrictions has only the empty restriction.
For $n\in\bbN$, the $n$-th level consists of restrictions chosen 
in the basic games for restriction $c$ on the $i$-th zone,
for all $c$ in the $(n-1)$-th level of the chosen restrictions,
and all $i$ such that 
$c$ does not restrict any positions in $I_i$.

Let $Q_i$ be
the upper bound on the length of an $(N_i,\phi_i)$-sequence for
basic game on the $I_i$-th zone 
(propositions \ref{Ccount},\ref{choice_depth}).
Suppose the zones are picked in such a way that
\begin{equation}\label{cond_nullset}
Q_i\leq 2^{|I_i|-i-1}
\end{equation}
Then by proposition \ref{chosen_size},
the sum of sizes of the sets of sequences chosen in the basic game
for $c$ on the $i$-th zone is at most $2^{-i-1}\lambda([c])$,
and summing the sizes over all zones, $1/2\lambda([c])$.
This implies that the size of the set of sequences
consistent with the $n$-th level of chosen restrictions is less than $2^{-n}$,
as it should be for the $n$-th level of a ML-test.

We'll say that the bound on the number of splits in the $i$-th zone was violated for a sequence $\sigma$ if 
the restriction,
 that restricts all of the positions to the same bits as $\sigma$,
 was $I_i$-split on more than $N_i$ times.

Suppose the zones are picked in such a way that
\begin{equation}\label{cond_N}
N_i\leq |I_i|-\log\phi_i-i
\end{equation}
Let's call a sequence on which the betting strategies achieve only low
($\leq2$) capital a sequence with low capital.
By proposition \ref{chosen_low_cap},
for large enough $i$,
if a restriction $z$ has a large ($>0$) subset
of sequences with low capital,
and the size, of the set of sequences for which the bound on the number of splits in the $i$-th zone was violated, is small enough,
then one of the chosen restrictions 
in the basic game for $z$ on the $i$-th zone, $z'$,
has a large subset of sequences with low capital.

Let $\epsilon_i$ denote the size of the set of sequences for which the bound on the number of splits in the $i$-th zone was violated.
If
\begin{equation}\label{cond_viol_size}
\lim_i \epsilon_i = 0
\end{equation}
then for every restriction $z$, if $[z]$  has a large subset
of sequences with low capital, there is some $i$ and a restriction $z'$, chosen in the basic game for $z$ on the $i$-th zone,
such that
$[z']$  has a large subset
of sequences with low capital.

Note that the set of sequences with low capital is large and consistent with the empty restriction.
By induction, there is some sequence of restrictions 
$z_1\prefixof z_2\prefixof\dots$ such that $z_n$ is in the
$n$-th level of chosen restrictions, and $[z_n]$
has a (large) subset of sequences with low capital.
By compactness, the set of sequences consistent with all of the restrictions $z_1,z_2,\dots$ contains a sequence with low capital.

We have shown that if the zones satisfy conditions
\ref{cond_nullset} and \ref{cond_N}, and $\BS_A,\BS_B$ satisfy
condition \ref{cond_viol_size} then
the set of sequences consistent with
the $n$-th level of chosen restrictions is an $n$-th level of a ML-test,
and there is a sequence which fails every level of this ML-test
on which neither betting strategy wins.   

\begin{definition}
	Let $I$ be a finite set of positions, and $(N, \phi)$ a pair of natural numbers.
	We will call the pair $(I,(N,\phi))$ a \textit{zone}.
	A sequence of zones $(I_1,N_1,\phi_1),(I_2,N_2,\phi_2),\dots$
	with disjoint sets of positions is called \textit{game zones}.
\end{definition}
\begin{definition}
Let $Z=(I_1,N_1,\phi_1),(I_2,N_2,\phi_2),\dots$ be some game zones.
If a betting strategy $I_i$-splits on a sequence more than $N_i$ times,
we'll say that \textit{the bound on the number of splits in the $i$-th zone was violated by the betting strategy} for this sequence.

Let $\BS_A,\BS_B$ be a pair of betting strategies.
When $\BS_A,\BS_B$ are known, we'll say that a sequence has low capital if both $\BS_A,\BS_B$ achieve capital on the sequence that is less than or equal to the treshold $2$.
	
	Let $z$ be a restriction that restricts a finite set of positions $P$.
	If $z$ does not restrict positions in $I_i$,
	let $C_i$ be the choice function on $I_i,z$ against $\BS_A,\BS_B$
	with parameters $N_i,\phi_i$.
	If $z$ does restrict positions in $I_i$,
	let $C_i$ be undefined on all inputs $m\in\bbN^*$, 
	that is, the set of restrictions chosen by $C_i$ is empty.
	
	We'll call the union over $i$ of restrictions chosen by $C_i$
	\textit{the restrictions chosen against $\BS_A,\BS_B$ on zones $Z$ for restriction $z$}.	
	We'll call the set of sequences consistent with those restrictions \textit{the sequences chosen against $\BS_A,\BS_B$ on zones $Z$ for restriction $z$}.	
\end{definition}
\begin{lemma}\label{lemma_cond_Q_size}
	Let $\BS_A,\BS_B$ be a pair of betting strategies and
	$Z=(I_1,(N_1,\phi_1)),(I_2,(N_2,\phi_2)),\dots$ be some game zones.
	Let $Q_i$ be the upper bound on the length of an 
	$(N_i,\phi_i)$-sequence against a pair of partition evaluations of a set with $2^{|I_i|}$ elements.
	Let $z$ be a restriction that restricts a finite set of positions.
	
	If $Q_i\leq 2^{|I_i|-i-1}$
	then the sum of sizes of the sets of sequences consistent with restrictions chosen against $\BS_A,\BS_B$ on zones $Z$ for restriction $z$ is less than $\frac{1}{2}\lambda([z])$.
\end{lemma}
\begin{proof}
	By proposition \ref{chosen_size}, on the $i$-th zone, the size of the set of chosen sequences is at most 
	$Q_i2^{-|I_i|}\lambda([z])\leq 2^{-i-1}\lambda([z])$.
	Summing over all $i$ we get the result.
\end{proof}

\begin{lemma}\label{cond_BS-bound_zone_continuable}
	Let $\BS_A,\BS_B$ be a pair of betting strategies.
	Let $Z=(I_1,(N_1,\phi_1)),(I_2,(N_2,\phi_2)),\dots$ be some game zones
	and $z$ some restriction that restricts a finite set of positions.
	
	Denote with $\epsilon_i$ the size of the set of sequences
	for which the bound on the number of splits in the $i$-th zone was violated by $\BS_A$ or $\BS_B$. 

	If the size
	of the set of sequences with low capital consistent with $z$ is larger than $0$, and
	\begin{enumerate}[label=(\roman*)]
		\item $\epsilon_i$ goes to zero as $i$ goes to infinity, and
		\item $N_i\leq |I_i|-\log\phi_i - i$
	\end{enumerate}
	then
	there is a restriction $z'$
	chosen against $\BS_A,\BS_B$ on zones $Z$ for restriction $z$
	such that
	the size
	of the set of sequences with low capital consistent with $z'$ is larger than $0$.
\end{lemma}
\begin{proof}
	Let $\theta$ be the size, conditional on $[z]$, of the set of sequences with low capital.
	We have that $\theta>0$.
	For large enough $i$, the value $\epsilon_i/\lambda([z])+2^{-i}$ 
	becomes arbitrarily small and
	by proposition \ref{chosen_low_cap}
	if for some $i$, $\theta-\epsilon_i/\lambda([z])-2^{-i}$ is larger than $0$,
	then the the choice function on $I_i,z$ against $\BS_A,\BS_B$
	with parameters $N_i,\phi_i$ chooses a restriction $[z']$
	such that
	the size, conditional on $[z']$,
	of the set of sequences with low capital is larger than $0$.
	Equivalently, the set of sequences with low capital consistent with $z'$ is larger than $0$.
\end{proof}

\begin{definition}
	Let $\BS_A,\BS_B$ be a pair of betting strategies and
	let $Z=(I_1,(N_1,\phi_1)),(I_2,(N_2,\phi_2)),\dots$ be some game zones.
	For $n\in\bbN$, we recursively define sets of restrictions $L_n$.
	Let $L_1$ be the set of restrictions
	chosen against $\BS_A,\BS_B$ on zones $Z$ for the empty restriction.
	Let $L_{n+1}$ be the union over restrictions $z\in L_n$ of 
	the set of restrictions
	chosen against $\BS_A,\BS_B$ on zones $Z$ for restriction $z$.

	We'll call the set of sequences that are, for all $n$, consistent with some restriction in $L_n$
	\textit{the chosen sequences against $\BS_A,\BS_B$ on zones $Z$}.
	We'll call $L_n$ \textit{the $n$-th level} of chosen restrictions
	against $\BS_A,\BS_B$ on zones $Z$.
\end{definition}

\begin{proposition}\label{cond_zones,BSs_Nonuniv}
	Let $\BS_A=(S_A,\mu_A),\BS_B=(S_B,\mu_B)$ be a pair of computable betting strategies with $\mu_A(\Omega)+\mu_B(\Omega)=1$, 
	where $\Omega=\bseqs$
	.
	Let $Z=(I_1,(N_1,\phi_1)),(I_2,(N_2,\phi_2)),\dots$ be some game zones
	.
	
	Denote with $Q_i$ the upper bound on the length of an 
	$(N_i,\phi_i)$-sequence against a pair of partition evaluations of a set with $2^{|I_i|}$ elements.
	
	Denote with $\epsilon_i$ the size of the set of sequences
	for which the bound on the number of splits in the $i$-th zone was violated. 
	
	If
	\begin{enumerate}[label=(\Roman*)]
		\item\label{cond_Q_size} 
		$Q_i\leq 2^{|I_i|-i-1}$, and
		\item\label{cond_viol_size_prop}
		$\epsilon_i$ goes to zero as $i$ goes to infinity, and
		\item\label{cond_N<I-logphi-i} 
		$N_i\leq |I_i|-\log\phi_i - i$
	\end{enumerate}
	then
	there is a non-MLR sequence on which neither of
	$\BS_A,\BS_B$ wins.
\end{proposition}
\begin{proof}
	By lemma \ref{lemma_cond_Q_size}, the sum of sizes of sets of sequences consistent with restrictions in the first level of chosen restrictions
	(against $\BS_A,\BS_B$ on zones $Z$) has size at most $1/2$.
	Suppose that the sum of sizes of sets of sequences consistent with restrictions in $n$-th level of the chosen restrictions
	has size at most $2^{-n}$.
	Again, by lemma \ref{lemma_cond_Q_size} the sum of sizes of sets of sequences consistent with restrictions chosen in $(n+1)$-th level has size at most $2^{-n-1}$.
	By induction, for all $n$, the (open) set of sequences consistent with restrictions in the $n$-th level of chosen restrictions
	has size at most $2^{-n}$. 
	Let this set be the $n$-th level of a ML-test.
	The chosen sequences fail every level of this ML-test, and are therefore
	non-MLR.
	
	Denote the set of sequences on which the betting strategies achieve capital less than $2$ (sequences with low capital) with $D$.
	From $\mu_A(\Omega)+\mu_B(\Omega)=1$ the size of $D$ is at least $1/2$.
	Since every sequence is consistent with the empty restriction,
	the size of the set of sequences with low capital consistent with the empty restriction is therefore larger than $0$.
	By lemma \ref{cond_BS-bound_zone_continuable}
	there is a restriction in the first level of chosen restrictions, 
	$z_1$, such that the subset of sequences with low capital consistent with $z_1$ is larger than $0$.
	Suppose there are some restrictions 
	$z_1\prefixeqof\dots\prefixeqof z_n$ such that $z_i$ is in
	the $i$-th level of chosen restrictions and $[z_n]$ has a subset
	of sequences with low capital larger than $0$.
	Again by lemma \ref{cond_BS-bound_zone_continuable},
	there is an extension of $z_n$ in the $n+1$-th
	level of the chosen restrictions, $z_{n+1}$ with a subset
	of sequences with low capital larger than $0$.
	By induction,
	there is a sequence of restrictions 
	$z_1\prefixeqof z_2\prefixeqof\dots$
	such that all of the sequences in $\bigcap_n [z_n]$ are chosen sequences,
	and for every $n$, $\lambda([z_n]\cap D)>0$.
	Since $D$ is a closed set, the set $[z_n]\cap D$ is closed,
	and (by compactness) the intersection of nonempty closed sets
	$\bigcap_n [z_n]\cap D$ is non empty and therefore
	$\bigcap_n [z_n]$ contains sequences with low capital.
\end{proof}


Suppose that for a betting strategy there is a
computable sequence of positions $\pi=p_1,p_2,\dots$ and an
unbounded computable function $f$,
such that, for an infinite binary sequence $\sigma$ it is almost surely true (w.r.t. $\lambda$) that
for all but finitely many $\ell$,
the betting strategy $\{p_1,\dots,p_\ell\}$-split on $\sigma$ at most $f(\ell)$ times.
We'll say that \textit{the betting strategy on $\pi$ has splits upper bounded by $f$}.

\begin{definition}\label{def_splits_bounded_by_f}
	Let 
	$\BS=(S,\mu)$ be a betting strategy.
	Let 
	$\pi=p_1,p_2,\dots$ be a computable sequence of distinct positions and
	$f$ some computable function.
	Denote with
	$\rho^\sigma_\ell$ a restriction consistent with a sequence $\sigma$ that restricts all except the first $\ell$ positions in $\pi$.

	Let $X$ be the set of sequences
	such that 
	for any sequence $\sigma\in X$  
	for all but finitely many $\ell$,
	for any part of the partition refinement $S$ that contains the sequence,
	the part was split on $\rho^\sigma_\ell$
	at most $f(\ell)$ times.
	If $X$ has size $1$, we'll say that \textit{$\BS$ on $\pi$ has splits upper bounded by $f$}.
	
	Let $X$ be the set of sequences
	such that 
	for any sequence $\sigma\in X$  
	for all but finitely many $\ell$,
	for any part of the partition refinement $S$ that contains the sequence,
	the part was split on $\rho^\sigma_\ell$
	at least $f(\ell)$ times.
	If $X$ has size $1$, we'll say that \textit{$\BS$ on $\pi$ has splits lower bounded by $f$}.
\end{definition}

\subsection{Game zones for bounded splits}
We show that if both betting strategies, on some computable sequence of positions $\pi$, have splits upper bounded by $f(\ell)=\ell-\log\ell -g(\ell)$,
where $g$, called the gap function, is unbounded and computable,
then we can construct game zones, called the zones on positions $\pi$ with gap $g$, that have the properties required by proposition \ref{cond_zones,BSs_Nonuniv}.
This implies there is a non-MLR sequence on which neither betting strategy wins.

\begin{definition}\label{def_zones}
	Let 
	$\pi$ be a computable sequence of distinct positions, and
	$g$ some unbounded computable function.
	We partition the positions $\pi$ into smallest consecutive intervals
	$I_1,I_2,\dots$ such that
	$g(\sum_{i=1}^{k}|I_i|)
	\geq 
	2k+2+\sum_{i=1}^{k-1}|I_i|$.
	For all $k$, let $\phi_k=2^{k+2}|I_k|$ and $N_k=|I_k|-\lfloor\log\phi_k\rfloor-k$. 
	We call the sequence of zones
	$(I_1,N_1,\phi_1),(I_2,N_2,\phi_2),\dots$ 
	\textit{the zones on positions $\pi$ with gap $g$}.
\end{definition}
\begin{lemma}\label{zone_chosen_small}
	Let $(I_i,N_i,\phi_i)$ be the $i$-th zone on positions $\pi$ with gap $g$.
	The upper bound on the length of an 
$(N_i,\phi_i)$-sequence against a pair of partition evaluations of a set with $2^{|I_i|}$ elements is less than $2^{|I_i|-i-1}$.	
\end{lemma}
\begin{proof}
	By proposition \ref{Ccount},
	The upper bound on the length of an 
	$(N_i,\phi_i)$-sequence against a pair of partition evaluations of a set with $2^{|I_i|}$ elements is
	$\frac{N_i+1}{\phi_i}2^{|I_i|+1}
	=
	\frac{|I_i|-\lfloor\log\phi_i\rfloor-i+1}{\phi_i}2^{|I_i|+1}
	\leq
	\frac{|I_i|}{\phi_i}2^{|I_i|+1}
	=
	\frac{|I_i|}{2^{i+2}|I_i|}2^{|I_i|+1}
	=
	2^{|I_i|-i-1}
	$
\end{proof}

\begin{lemma}\label{zone_f_bound}
	Let $\pi$ be a computable sequence of positions, and 
	$\BS_A,\BS_B$ be a pair of betting strategies
	that have splits on $\pi$ upper bounded by some $f$.
	
	Let $(I_1,N_1,\phi_1),(I_2,N_2,\phi_2),\dots$ be some game zones
	where $I_1,I_2,\dots$ are consecutive intervals of $\pi$.
	
	If
	$f(\sum_{i=1}^k|I_i|)\leq N_k$ 
	then the size of the set of sequences for which the bound on the number of splits in the $i$-th zone was violated goes to $0$
	when $i$ goes to infinity.
\end{lemma}
\begin{proof}
	Let $P_k$ the union of the intervals of positions of the first $k$ zones
	and let $\ell_k$ be the size of $P_k$, 
	that is,
	$\ell_k=\sum_{i=1}^k|I_i|$.

	If a betting strategy $P_k$-splits on a sequence $n$ many times, then it $I_k$-splits on a sequence at most $n$ many times, since
	$I_n\subseteq P_k$.
	Then from definition \ref{def_splits_bounded_by_f},
	the size of the set of sequences on which either of the
	betting strategies $P_k$-split more than $f(\ell_k)$ times
	goes to zero as $k$ goes to infinity.
	But then also, 	the size of the set of sequences on which either of the
	betting strategies $I_k$-split more than $N_k$ times
	goes to zero as $k$ goes to infinity.
	In other words,
	the size of the set of sequences for which the bound on the number of splits in the $k$-th zone was violated goes to $0$
	when $k$ goes to infinity.
\end{proof}
\begin{lemma}\label{zone_N_bound}
	Let $(I_1,N_1,\phi_1),(I_2,N_2,\phi_2),\dots$ be the zones on positions $\pi$ with gap $g$.
	Let $\BS_A,\BS_B$ be a pair of betting strategies
	that have splits on $\pi$ upper bounded by $f(\ell)=\ell-\log\ell-g(\ell)$.
	
	The size of the set of sequences for which the bound on the number of splits in the $i$-th zone was violated goes to $0$
	when $i$ goes to infinity.
\end{lemma}
\begin{proof}
	Let $\ell_k$ be the length of the initial segment of positions in $\pi$ contained in the intervals of positions of the first $k$ zones, that is,
	$\ell_k=\sum_{i=1}^k|I_i|$.
	
	$\ell_k-\log\ell_k-g(\ell_k)
	\leq
	\sum_{i=1}^k|I_i|-\log(\sum_{i=1}^k|I_i|)-\sum_{i=1}^{k-1}|I_i|- 2k -2
	\leq
	|I_k|-\log|I_k|-2k-2
	=
	|I_k|-(k+2+\log|I_k|)-k
	=
	|I_k|-\log\phi_k-k
	\leq N_k$
	.
	
	By definition \ref{def_zones}, $I_1,I_2,\dots$ are consecutive intervals of $\pi$ and the result follows from
	lemma \ref{zone_f_bound}.
\end{proof}

\begin{theorem}\label{th_upper_bound}
	Let $\BS_A=(S_A,\mu_A),\BS_B=(S_B,\mu_B)$ be a pair of computable betting strategies with $\mu_A(\Omega)+\mu_B(\Omega)=1$
	.
	Let $\pi$ be a computable sequence of distinct positions,
	$g$ an unbounded computable function and
	let $Z$ be the zones on positions $\pi$ with gap $g$.
	
	If both $\BS_A,\BS_B$ on $\pi$ have splits upper bounded by
	$\ell-\log\ell-g(\ell)$,
	there is a non-Martin-L\"of random sequence on which neither strategy wins.
\end{theorem}
\begin{proof}
	By lemmas \ref{zone_chosen_small}, 
	\ref{zone_N_bound},
	and definition \ref{def_zones}
	the conditions
	 \ref{cond_Q_size},
	\ref{cond_viol_size_prop}, 
	\ref{cond_N<I-logphi-i}
	of proposition \ref{cond_zones,BSs_Nonuniv} are fullfilled
	and the result follows.
\end{proof}

\subsection{Half-splitting game zones for upper bounded splits}
For a pair of half-splitting betting strategies, 
the length of an $(N,\phi)$-sequence 
against a pair of partition evaluations of a set with $2^{|I|}$ elements
depends only on $\phi$ (proposition \ref{Ccount}).
This allows us to get a better bound on the number of splits.
Suppose that the betting strategies,
on some computable sequence of positions $\pi$, have splits upper bounded by $f(\ell)=\ell -g(\ell)$,
where $g$ is unbounded and computable,
then we can construct game zones, called the half-splitting zones on positions $\pi$ with gap $g$, that have the properties required by proposition \ref{cond_zones,BSs_Nonuniv}.
This implies there is a non-MLR sequence on which neither betting strategy wins.

\begin{definition}\label{def_half-splitting_zones}
	Let 
	$\pi$ be a computable sequence of distinct positions, and
	$g$ some unbounded computable function.
	We partition the positions $\pi$ into smallest consecutive intervals
	$I_1,I_2,\dots$ such that
	$g(\sum_{i=1}^{k}|I_i|)\geq 2k +4+\sum_{i=1}^{k-1}|I_i|$.
	For all $k$, let $\phi_k=6\cdot2^{k+1}$ and $N_k=|I_k|-2k-4$. 
	We call the sequence of zones
	$(I_1,N_1,\phi_1),(I_2,N_2,\phi_2),\dots$ 
	\textit{the half-splitting zones on positions $\pi$ with gap $g$}.
\end{definition}
\begin{lemma}\label{half-betting_zone_chosen_small}
	Let $(I_i,N_i,\phi_i)$ be the $i$-th half-betting zone on positions $\pi$ with gap $g$.
	The upper bound on the length of an 
	$(N_i,\phi_i)$-sequence against a pair of partition evaluations of a set with $2^{|I_i|}$ elements is less than $2^{|I_i|-i-1}$.	
\end{lemma}
\begin{proof}
	By definition \ref{def_half-splitting_zones}, $\phi_i=6\cdot2^{i+1}$ and the result follows from proposition \ref{Ccount}.
\end{proof}
\begin{lemma}\label{half-splitting_N_k<I-logphi_k-k}
	Let $(I_k,N_k,\phi_k)$ be the $k$-th half-splitting zone on positions $\pi$ with gap $g$.
	We have 
	$N_k\leq |I_k|-\log\phi_k-k$.
\end{lemma}
\begin{proof}
	$|I_k|-\log\phi_k-k
	=
	|I_k|-\log6\cdot2^{k+1}-k
	=
	|I_k|-\log6 -2k-1
	\geq
	|I_k|-2k-4
	=
	N_k
	$
\end{proof}
\begin{lemma}\label{zone_N_bound_half-splitting}
	Let $(I_1,N_1,\phi_1),(I_2,N_2,\phi_2),\dots$ be the 
	half-splitting zones on positions $\pi$ with gap $g$.
	Let $\BS_A,\BS_B$ be a pair of half-splitting betting strategies
that have splits on $\pi$ upper bounded by $f(\ell)=\ell-g(\ell)$.

The size of the set of sequences for which the bound on the number of splits in the $i$-th zone was violated goes to $0$
when $i$ goes to infinity.
\end{lemma}
\begin{proof}
	Let $\ell_k$ be the length of the initial segment of positions in $\pi$ contained in the intervals of positions of the first $k$ zones, that is,
	$\ell_k=\sum_{i=1}^k|I_i|$.
	
	$\ell_k-g(\ell_k)
	\leq
	\sum_{i=1}^k|I_k|-\sum_{i=1}^{k-1}|I_i|- 2k -4
	=
	|I_k|-2k-4
	=
	N_k
	$
	
	By definition \ref{def_half-splitting_zones}, $I_1,I_2,\dots$ are consecutive intervals of $\pi$ and the result follows from
	lemma \ref{zone_f_bound}.
\end{proof}

\begin{theorem}\label{th_upper_bound_half-splitting}
	Let $\BS_A,\BS_B$ be a pair of betting strategies that are half-splitting.
	Let $\pi$ be a computable sequence of distinct positions,
	$g$ an unbounded computable function and
	let $Z$ be the half-splitting zones on positions $\pi$ with gap $g$.
	
	If both $\BS_A,\BS_B$ on $\pi$ have splits upper bounded by
	$\ell-g(\ell)$,
	there is a non-Martin-L\"of random sequence on which neither strategy wins.
\end{theorem}
\begin{proof}
	By lemmas \ref{half-betting_zone_chosen_small}, 
	\ref{zone_N_bound_half-splitting},
	\ref{half-splitting_N_k<I-logphi_k-k}
	the conditions
	\ref{cond_Q_size},
	\ref{cond_viol_size_prop}, 
	\ref{cond_N<I-logphi-i}
	of proposition \ref{cond_zones,BSs_Nonuniv} are fullfilled
	and the result follows.
\end{proof}

\section{ML-test for a pair of KLBS-es with lower bounded splits}

For a pair of Kolmogorov-Loveland betting strategies that have
splits lower bounded by $\ell-g(\ell)$, where $g$ is sublinear,
we have that for any
two constants $c,\epsilon$ and large enough $\ell,t$,
the set of sequences that were $[1,\ell]$-split on
at least $\ell(1-c)$ many times by
both KLBS-es up to time $t$,
has size at least $1-\epsilon$.
This implies that among first $\ell$ positions,
there is a position $p$ such that
the set of sequences that were $\{p\}$-split on
(once) by
both KLBS-es up to time $t$,
has size at least $(1-\epsilon)(1-c)\geq 1-\epsilon-c$.
Furthermore, for any sequence of bounds $\xi_1,\xi_2,\dots$,
we can find positions $\pi=p_1,p_2,\dots$ and times $t_1,t_2,\dots$
such that
the set of sequences that were $\{p_i\}$-split on
by
both KLBS-es up to time $t_i$,
has size at least $1-\xi_i$.

We can pick small enough bounds $\xi_1,\xi_2,\dots$
so that for any $n$ and any restriction $r$
that restricts first $n$ positions in $\pi$,
there is some $t_n$, such that the size, conditional on $[r]$,
of the set of sequences that were $\{p_i\}$-split on by both KLBS-es up to time $t_i$, for all $i\leq n$ is larger than $1/2$.

We construct an infinite restriction $\zeta$ that restricts the positions in $\pi$ and contains a sequence on which neither KLBS wins.
The proof is somewhat similar to the proofs that permutation random sequences
are the same as computably random sequences \cite{Lempp}, \cite{Merkle}, \cite{Laurent}, which use tools introduced in \cite{BPP_EXP}.
The difference is that, instead of always betting on all positions in some prescribed order (as in the definition of permutation randomness in \cite{Open}), here we'll have that a KLBS bets on almost all positions in $\pi$ almost surely, and can do so adaptively.
 
Suppose a finite restriction $z_{n-1}$ that restricts the first $n-1$ positions in $\pi$ is already defined.
Denote  with $q_{n-1}$ the sum of masses over a subset of parts contained
in $[z_{n-1}]$,
that have split on position $p_i$ by the time $t_i$, for all $i\leq n-1$.

The restriction $z_n$ extends $z_{n-1}$ by restricting the position
$p_n$ to a value that minimizes $q_n$, the sum of masses over
a subset of parts contained in $[z_{n}]$,
that have split on position $p_i$ by the time $t_i$, for all $i\leq n$.
This implies that $q_n$ is less than $q_{n-1}/2$, and by induction 
is for all $n$ less than $2^{-n}$, that is, the size of $[z_n]$.

On the other hand, for all $n$, the size, conditional on $[z_n]$,
of the set of sequences that were, for all $i\leq n$,
$\{p_i\}$-split on by both KLBS-es up to time $t_i$,  
is larger than $1/2$.
This implies the set $[z_n]$ at all times $t$ contains a sequence
on which the capital of terminal parts that contain the sequence is below
$2$, as otherwise $q_n$ would be larger than $\lambda([z_n])$.

This is still not enough to prove the claim that the restriction $z_0\prefixeqof z_1\prefixeqof\dots\zeta$ contains a sequence on which neither KLBS wins, since we need to look at the maximum over all $t$ of capital 
of a terminal part that contains the sequence.
To remedy this, we use the "savings trick",
\begin{proposition}["slow-but-sure winnings" lemma in \cite{BPP_EXP}]\label{savings_trick}
For any given betting strategy $\BS=(S,\mu)$, we can construct 
$\BS'=(S,\mu')$ that
wins on every sequence on which $\BS$ wins,
and the difference between the capital and the maximal capital of a part is bounded. 
\end{proposition}
If the difference between the capital and the maximal capital of a part is bounded, then the set $[z_n]$, at all times $t$, contains a sequence
on which the betting strategies achieve bounded capital. That is, the KLBS pair doesn't win on some sequence in $[z_n]$, for all $n$. But then $[\zeta]$ also contains a sequence on which the KLBS-es do not win. Since $[\zeta]$ is an effective nullset, all of the sequences in it are non-MLR.

\begin{theorem}\label{th_lower_bound}
	If a pair of computable KLBS-es has splits, on some sequence of positions $\pi$,
	lower bounded by $\ell-o(\ell)$,
	then there is a non-MLR sequence on which neither strategy wins.
\end{theorem}
\begin{proof}
	\begin{claim}\label{positions_read}
		There must be a sequence of position $p_1,p_2,\dots$
		and times $t_1, t_2,\dots$ such that
		the size of
		the set of sequences that were $\{p_i\}$-split on
		by both KLBS-es up to time $t_i$ is at least $1-2^{-2i}$
	\end{claim}
	\begin{proof}
		Suppose the claim is not true:
		there is some bound $d$ such that
		for every position $p$, at all times $t$,
		the size of
		the set of sequences that were $\{p\}$-split on
		by both KLBS-es, up to time $t$, is at most $1-d$.
		
		Let $x$ denote the sum, over the first $\ell$ positions $p$ in $\pi$, 
		of the sizes of sets of sequences on which
		at least one KLBS from the pair did not $\{p\}$-split on.
		The sum $x$ is at least $d\ell$.
		
		On the other hand, let $\epsilon,c$ be such that
		the size of the set of sequences on which either of the betting strategies has split less than $\ell-c\ell$ many times up to time $t$
		is less than $\epsilon$.
		The size of the set of sequences on which the betting strategies
		do not $\{p\}$-split on, for every $p$ among the first $\ell$ positions in $\pi$ is at most $\epsilon$,
		and for the remaining sequences there are at most $2c\ell$ positions $p$ such that one of the two strategies did not $\{p\}$-split on the sequence. 
		Therefore, the sum $x$ is at most $\epsilon\ell+(1-\epsilon)(2c\ell)$.
		
		We have that 
		$d\ell
		\leq
		\epsilon\ell+(1-\epsilon)(2c\ell)
		\leq
		(\epsilon+2c)\ell
		$.
		
		This is in contradiction with the assumption of the lemma, that on positions $\pi$,
		for every pair of constants $\epsilon,c$ for large enough $\ell,t$,
		the size of the set of sequences on which either of the betting strategies has split less than $\ell-c\ell$ many times up to time $t$
		is less than $\epsilon$.
	\end{proof}
	Note that the sequence of positions $p_1,p_2,\dots$, and times $t_1,t_2,\dots$
	in the claim \ref{positions_read} can be found effectively.
	But then, we can also effectively find a subsequence
	$p'_1,p'_2,\dots$ and $t'_1,t'_2,\dots$ such that
	$t'_1<t'_2<\dots$ and the size of
	the set of sequences that were $\{p'_i\}$-split on
	by both KLBS-es up to time $t'_i$ is at least $1-2^{-2i}$.

	By proposition \ref{savings_trick}, it is enough to consider only
	betting strategy with the savings property:
	the difference between capital and maximum capital of a part is bounded. 
	Let $\BS_A=(S_A,\mu_A),\BS_B=(S_B,\mu_B)$ denote the pair of KLBS-es with
	the savings property, and $\mu_A(\{0,1\}^\infty)+\mu_B(\{0,1\}^\infty)=1$. 
	We will construct a sequence of restrictions
	$z_0,z_1,z_2\dots$ such that
	$z_n$ extends $z_{n-1}$ by restricting the position $p'_n$,
	and for each $n$, $[z_n]$ contains a sequence on which both strategies achieve capital below some treshold.
	
	Let $z_0$ be the empty restriction,
	and let $V_A^0,V_B^0$ be the subset of parts, terminal at $t_0=0$ w.r.t. $S_A,S_B$, that are contained in $[z_0]$ (namely, both $V_A^0,V_B^0$ contain the only part that is terminal at $0$, that is, the set of all sequences). 
	When sets $V_A^n,V_B^n$ are defined,
	denote with $W^n$
	the union of sequences contained in intersections of parts in $V_A^n,V_B^n$,
	that is, $W^n=(\bigcup_{v\in V^n_A}v)\cap(\bigcup_{v\in V^n_A}v)$.
	
	Suppose $z_{n-1},V_A^{n-1},V_B^{n-1}$ are already defined,
	and $\lambda(W^{n-1})\geq \frac{1}{2}(1+2^{-(n-1)})\lambda([z_{n-1}])$.
	The restriction $z_n$ restricts position $p'_n$ to the value that minimizes 
	the sum of masses of both strategies assigned to the parts,
	terminal at $t'_n$ w.r.t. $S_A,S_B$, that are contained in $[z_n]$ and are a subset of part in $V^{n-1}_A,V^{n-1}_B$.
	Denote these subsets of parts with $V^n_A,V^n_B$.

	Since the value assigned to $p'_n$ by $z_n$ minimizes the sum,
	we have that 
	$\mu_A(V^n_A)+\mu_B(V^n_B)
	\leq\frac{1}{2}(\mu_A(V^{n-1}_A)+\mu_B(V^{n-1}_B))$.
	
	The size of the set of sequences that were $\{p'_n\}$-split on
	by both KLBS-es up to time $t'_n$ is at least $1-2^{-2n}$,
	implying that the size of the set of sequences on which at least one strategy, up to time $t'_n$, did not $\{p'_n\}$-split on is at most $2^{-2n}$.
	But then, 
	$
	\lambda(W^n)
	\geq
	\frac{1}{2}(\lambda(W^{n-1})-2^{-2n})
	\geq
	\frac{1}{2}(\frac{1}{2}(1+2^{-(n-1)})2^{-(n-1)}-2^{-2n})
	=
	\frac{1}{2}(2^{-n}+2^{-2n+1}-2^{-2n})
	=
	\frac{1}{2}(1+2^{-n})\lambda([z_n])
	$.
	
	By induction, for all $n$, the size of $W^n$ is more than half the size of $[z_n]$, and the sum of masses of parts that contain sequences from $W^n$ is at most $2^{-n}(\mu_A(\{0,1\}^\infty)+\mu_B(\{0,1\}^\infty))=2^{-n}$.
	Then, for every $t$ there must be a sequence in $W^n$ contained in
	intersection of parts $a,b$, terminal at $t$, with capital less than $2$.
	Since the betting strategies have the savings property,
	this implies that for all $n$, $[z_n]$ contains a sequence on which
	the betting strategies achieve capital below some treshold.
	The set of sequences on which
	the betting strategies achieve capital below the treshold is closed,
	and by compactness, the set $\bigcap_n[z_n]$ contains  a sequence on which
	the betting strategies achieve capital below the treshold (they do not win on the sequence).
	The set $\bigcap_n[z_n]$ is an effective nullset and all of the sequences in it are non-MLR. 	
\end{proof}

\end{document}